\documentclass[12pt]{article}
\usepackage{latexsym,amssymb,float,amsmath,amsthm,enumerate,geometry,algorithm2e, tikz}
\usepackage{verbatim}
\newtheorem{theorem}{Theorem}

\newtheorem{lemma}[theorem]{Lemma}

\usepackage{lineno}
\usepackage{setspace}
\allowdisplaybreaks

\usetikzlibrary{decorations.text}
\usetikzlibrary{decorations.pathmorphing}
\usepackage{pgfplots}
\pgfplotsset{compat=1.18}
\usetikzlibrary{calc}

\begin{document}
%\linenumbers
\onehalfspace

\title{Largest common subgraph of two forests}
\author{Dieter Rautenbach \and Florian Werner}
\date{}

\maketitle
\vspace{-10mm}
\begin{center}
Institute of Optimization and Operations Research, Ulm University, Ulm, Germany\\
\texttt{$\{$dieter.rautenbach,florian.werner$\}$@uni-ulm.de}
\end{center}

\medskip

$\mbox{}$\hfill ``{\it My god, it's full of stars}''\\
$\mbox{}$\hfill {\it --- David Bowman, 2001}

\vspace{-3mm}

\begin{abstract}
A common subgraph of two graphs $G_1$ and $G_2$ is a graph 
that is isomorphic to subgraphs of $G_1$ and $G_2$. 
In the largest common subgraph problem the task is to determine
a common subgraph for two given graphs $G_1$ and $G_2$
that is of maximum possible size ${\rm lcs}(G_1,G_2)$. 
This natural problem generalizes the well-studied graph isomorphism problem,
has many applications, and remains NP-hard even restricted to unions of paths.
We present a simple $4$-approximation algorithm for forests, and, 
for every fixed $\epsilon\in (0,1)$,
we show that,
for two given forests $F_1$ and $F_2$ of order at most $n$,
one can determine in polynomial time
a common subgraph $F$ of $F_1$ and $F_2$
with at least ${\rm lcs}(F_1,F_2)-\epsilon n$ edges.
Restricted to instances with ${\rm lcs}(F_1,F_2)\geq cn$ for some fixed positive $c$,
this yields a polynomial time approximation scheme.
Our approach relies on the approximation of the given forests
by structurally simpler forests 
that are composed of copies of only $O(\log (n))$ different starlike rooted trees
and iterative quantizations of the options for the solutions.\\[3mm]
\end{abstract}

\section{Introduction}

We consider finite, simple, and undirected graphs, and use standard terminology.
A graph $H$ is a {\it subgraph} of a graph $G$ 
if $H$ arises from $G$ by removing vertices and edges
and a subgraph of $G$ is {\it spanning}
if it contains all vertices of $G$.
If $H$ is a subgraph of a graph $G$, we write $H\subseteq G$.
The order $n(G)$ and the size $m(G)$ of a graph $G$ 
are the numbers of its vertices and edges of $G$, respectively.
For a positive integer $k$, 
let $[k]$ be the set of positive integers at most $k$ and 
let $[k]_0=\{ 0\}\cup [k]$.

In their seminal list of NP-complete problems, 
Garey and Johnson \cite{gajo} mention the following decision problem as [GT49].

\medskip

\noindent {\sc Largest Common Subgraph}\\
\begin{tabular}{lp{14cm}}
Instance: & Two graphs $G$ and $H$, and a positive integer $K$.\\
Question: & Are there spanning subgraphs $G'$ of $G$
and $H'$ of $H$ that have at least $K$ edges and are isomorphic?
\end{tabular}

\medskip

Note that two graphs $G$ and $H$ are isomorphic 
if and only if 
$G$, $H$, and $K=m(G)$ form a yes-instance of {\sc Largest Common Subgraph},
that is, this problem generalizes the graph isomorphism problem.
It was first proposed by Bokhari \cite{bo}
within the context of array processing
and has various applications ranging from molecular chemistry \cite{rawi}
to pattern matching \cite{shbuve}.
Since the subgraphs $G'$ and $H'$ in the problem statement 
are obtained by removing edges only,
{\sc Largest Common Subgraph} makes sense only 
for graphs $G$ and $H$ of the same order.
Its NP-completeness follows easily from 
the NP-completeness of {\sc Clique}, which is [GT19] in \cite{gajo}.

In their comment to [GT49],
Garey and Johnson claim that {\sc Largest Common Subgraph}
can be solved in polynomial time if both $G$ and $H$ are trees,
for which they cite a private communication by Edmonds and Matula from 1975.
Grohe, Rattan, and Woeginger \cite{grrawo} show that this claim is false
by a reduction from {\sc 3-partition}, which is [SP15] in \cite{gajo}.
In fact, let $I$ be an instance of {\sc 3-partition}
that consists of $3m$ positive integers $a_1,\ldots,a_{3m}$
with $A/4<a_i<A/2$ for each $i\in [3m]$,
where $A=\frac{1}{m}(a_1+\cdots+a_{3m})$.
Recall that the question for $I$ is whether there is a partition of $[3m]$
into $m$ sets $I_1,\ldots,I_m$ each containing exactly three elements
such that $\sum_{j\in I_i}a_j=A$ for each $i\in [m]$.
Now, it is easy to see that $I$ is a yes-instance of {\sc 3-partition}
if and only if $G$, $H$, and $K$ form a yes-instance of {\sc Largest Common Subgraph},
where $G$ is the disjoint union of $3m$ paths of orders $a_1,\ldots,a_{3m}$,
$H$ is the disjoint union of $m$ paths each of order $A$, and
$K$ is the size of $G$.
Since {\sc 3-partition} is NP-complete in the strong sense \cite{gajo},
it follows that {\sc Largest Common Subgraph} remains NP-complete
when restricted to instances where $G$ and $H$ are unions of paths.
A simple modification yields the following.

\medskip

\noindent {\bf Theorem} (Grohe, Rattan, and Woeginger, Theorem 8 in \cite{grrawo})
{\sc Largest Common Subgraph} {\it remains NP-complete
when restricted to instances where $G$ and $H$ are both trees.}

\medskip

Possibly, Edmonds and Matula had a different problem in mind;
namely to determine the largest common subtree of two given trees.
In \cite{ak} Akutsu gives details for 
a simple efficient dynamic programming algorithm solving this problem 
using the maximum weight bipartite matching algorithm as a subroutine.
Many variations of {\sc Largest Common Subgraph} have been considered 
in view of their relevance for certain applications.
The variations involve 
the restriction to connected common subgraphs,
vertex/edge labels that have to be respected, and 
topological notions of subgraphs,
cf. \cite{akta,akha,gakn,krkumu,yaaoma} and references therein.

In the present paper, we consider approximation algorithms for 
the maximization version of {\sc Largest Common Subgraph} restricted to forests.
Let a graph $H$ be a {\it common subgraph} of two graphs $G_1$ and $G_2$
if $H$ is isomorphic to a subgraph of $G_1$ as well as to a subgraph of $G_2$. 
Let ${\rm lcs}(G_1,G_2)$ be the largest size $m(H)$ 
of a common subgraph $H$ of $G_1$ and $G_2$.
Note that we ignore the restriction to spanning subgraphs 
from the statement of {\sc Largest Common Subgraph}
as it is not essential.
In fact, if $H$ is a common subgraph of two graphs $G_1$ and $G_2$,
which are both of the same order $n$, 
then adding $n-n(H)$ isolated vertices to $H$
yields a graph $H'$ of the same size as $H$
that is isomorphic to spanning subgraphs of $G_1$ and $G_2$.

While {\sc Largest Common Subgraph} remains NP-complete 
when restricted to unions of paths, 
we show with Lemma \ref{lemma1} below 
that it can be solved efficiently when restricted to unions of stars.
This yields a simple $4$-factor approximation algorithm.

\begin{theorem}\label{theorem1}
For two given forests $F_1$ and $F_2$ of order $n$, 
one can determine 
a common subgraph $F$ of $F_1$ and $F_2$ with 
$$m(F)\geq \frac{1}{4}{\rm lcs}(F_1,F_2)$$
in time $n^{O(1)}$.
\end{theorem}
A natural next goal would be a polynomial time approximation scheme (PTAS)
for {\sc Largest Common Subgraph} restricted to forests.
Our second result yields a PTAS when restricted to instances $(F_1,F_2)$,
where $F_1$ and $F_2$ are forests of order $n$ and 
${\rm lcs}(F_1,F_2)\geq c n$ for some fixed positive $c$.

\begin{theorem}\label{theorem2}
For every $\epsilon\in (0,1)$,
there is some $k\in \mathbb{N}$ with the following property:
For two given forests $F_1$ and $F_2$ of order $n$, 
one can determine 
a common subgraph $F$ of $F_1$ and $F_2$ with 
$$m(F)\geq {\rm lcs}(F_1,F_2)-\epsilon n$$
in time $O(n^k)$.
\end{theorem}
Our approach for Theorem \ref{theorem2} is as follows.
Firstly, removing a small fraction of all edges, 
we approximate the given forests $F_1$ and $F_2$ 
by simpler forests that are composed of copies of $O(\log n)$ 
different trees that are structurally close to stars.
In particular, each component $K$ of the approximating forests 
has a root $r$ of controlled degree and the components of $K-r$
are of bounded order.
Secondly, we show how to solve the largest common subgraph problem
approximately on such simpler instances.
For this approximate solution,
we reduce the number of options 
that need to be considered at several stages 
during our algorithm
by suitable quantization.

The following two sections contain proofs of our results and auxiliary statements.

\section{Proof of Theorem \ref{theorem1}}

In this section, we show Theorem \ref{theorem1}.

Our first lemma implies that, 
for two given unions $F_1$ and $F_2$ of stars,
a common subgraph $F$ of $F_1$ and $F_2$
of maximum size ${\rm lcs}(F_1,F_2)$
can be found efficiently.
The key observation is the following simple inequality:
For every $a,a',b,b'\in \mathbb{N}_0$ with $a<a'$ and $b<b'$, we have
\begin{eqnarray}\label{e1}
\min\{ a,b'\}+\min\{ a',b\}\leq \min\{ a,b\}+\min\{ a',b'\}.
\end{eqnarray}
(\ref{e1}) follows easily by considering all possible non-decreasing orderings of $a,b,a',b'$.

\begin{lemma}\label{lemma1}
Let $a_1,\ldots,a_\ell$ and $b_1,\ldots,b_\ell$ 
be two non-decreasing sequences of non-negative integers.
If $F_1$ is the disjoint union of $\ell$ stars of orders $a_1+1,\ldots,a_\ell+1$, and 
$F_2$ is the disjoint union of stars of orders $b_1+1,\ldots,b_\ell+1$, then
$${\rm lcs}(F_1,F_2)=\sum\limits_{i=1}^\ell\min\{ a_i,b_i\}.$$
\end{lemma}
\begin{proof}
Let $F$ be a common subgraph of $F_1$ and $F_2$ with $m(F)={\rm lcs}(F_1,F_2)$.
By renaming vertices, we may assume $F\subseteq F_1,F_2$.
Let $S_1,\ldots,S_\ell$ be the components of $F_1$,
where $S_i$ has order $a_i+1$, and 
let $T_1,\ldots,T_\ell$ be the components of $F_2$,
where $T_j$ has order $b_j+1$.
Let $H$ be the bipartite graph 
with the two partite sets $\{ S_1,\ldots,S_\ell\}$ and $\{ T_1,\ldots,T_\ell\}$,
where $S_i$ is adjacent to $T_j$ if and only if
some edge of $F$ belongs to $S_i$ as well as $T_j$.
Since all considered components are stars, 
the edges of $H$ form a matching $M$ in $H$.
Now, if $S_iT_j\in M$, then the choice of $F$ implies that 
$\min\{ a_i,b_j\}$ edges of $S_i$ and $T_j$ belong to $F$,
that is, $m(F)=\sum\limits_{S_iT_j\in M}\min\{ a_i,b_j\}$.
In view of this formula, we may assume that $M$ is a perfect matching of $H$.
In other words, there is a permutation $\pi$ of $[\ell]$ such that 
$m(F)=\sum\limits_{i=1}^\ell\min\{ a_i,b_{\pi(i)}\}$.
Now, (\ref{e1}) implies that choosing the permutation $\pi$ as the identity 
maximizes $\sum\limits_{i=1}^\ell\min\{ a_i,b_{\pi(i)}\}$,
which completes the proof.
\end{proof}
Now, Theorem \ref{theorem1} follows easily 
by decomposing the given forests into unions of stars.
Note that Lemma \ref{lemma1} implicitly assumes 
that $F_1$ and $F_2$ have equally many components, 
which can easily be ensured by adding isolated vertices. 

\begin{proof}[Proof of Theorem \ref{theorem1}]
Let $F$ be a common subgraph of $F_1$ and $F_2$ with $m(F)={\rm lcs}(F_1,F_2)$.
By renaming vertices, we may assume $F\subseteq F_1,F_2$.
For $i\in [2]$, let the set $R_i$ contain exactly one vertex from every component of $F_i$.
Recall that the distance of an edge $e$ from $R_i$ in $F_i$ 
is the minimum length of a path in $F_i$ intersecting both $e$ and $R_i$.
For $i\in [2]$, let $F_i^{\rm even}$ be the spanning subgraph of $F_i$
containing all edges of $F_i$ that have even distance to $R_i$,
and let $F_i^{\rm odd}=F_i-E\left(F_i^{\rm even}\right)$.
By construction, all components of 
$F_1^{\rm even}$,
$F_1^{\rm odd}$,
$F_2^{\rm even}$, and
$F_2^{\rm odd}$ are stars.
Furthermore, 
one of the four sets
$E\left(F_1^{\rm even}\right)\cap E\left(F_2^{\rm even}\right)$,
$E\left(F_1^{\rm even}\right)\cap E\left(F_2^{\rm odd}\right)$,
$E\left(F_1^{\rm odd}\right)\cap E\left(F_2^{\rm even}\right)$, and
$E\left(F_1^{\rm odd}\right)\cap E\left(F_2^{\rm odd}\right)$
contains at least $1/4$ of the edges of $F$.
Hence, efficiently determining four common subgraphs 
of maximum sizes for the pairs
$(F_1^{\rm even},F_2^{\rm even})$,
$(F_1^{\rm even},F_2^{\rm odd})$,
$(F_1^{\rm odd},F_2^{\rm even})$, and 
$(F_1^{\rm odd},F_2^{\rm odd})$
using Lemma \ref{lemma1},
and returning the one with most edges,
yields a common subgraph of $F_1$ and $F_2$ 
with at least $\frac{1}{4}{\rm lcs}(F_1,F_2)$ edges.
\end{proof}

\section{Proof of Theorem \ref{theorem2}}

In this section, we show Theorem \ref{theorem2}.

As one ingredient of the proof we need that a largest common subgraph
of two given forests with components of bounded orders
can be found efficiently 
by a straightforward dynamic programming approach;
the next lemma gives details.
For a positive integer $\Delta$,
let ${\cal F}_{\Delta}$
be the collection of all forests 
whose components have orders at most $\Delta$.

\begin{lemma}\label{lemma3}
For every $\Delta\in\mathbb{N}$, 
there is some $k\in\mathbb{N}$ with the following property:
For two given forests $F_1$ and $F_2$ of orders at most $n$ from ${\cal F}_{\Delta}$, 
one can determine a common subgraph $F$ of $F_1$ and $F_2$ 
with $m(F)={\rm lcs}(F_1,F_2)$
in time $O(n^k)$.
\end{lemma}
\begin{proof}
Let $\{ T_1,\ldots,T_p\}$ be the set of all (unrooted) trees of order at most $\Delta$,
in particular, $p$ is bounded in terms of $\Delta$.
For every forest $F$ of order $n$ from ${\cal F}_{\Delta}$,
there is a unique $k$-tupel $t(F)=(t_1,\ldots,t_p)\in [n]_0^p$ 
such that $F$ is isomorphic to the disjoint union of $t_i$ copies of $T_i$ for $i\in [p]$,
that is, $F\simeq\bigcup\limits_{i=1}^p t_iT_i$.
For a forest $F$ of order $n$ from ${\cal F}_{\Delta}$, 
note that every spanning subforest $F'$ of $F$ also belongs to ${\cal F}_{\Delta}$
and let 
$${\cal T}(F)
=\left\{
t(F'): \mbox{ $F'$ is a spanning subforest of $F$}
\right\}\subseteq [n]_0^p.$$
Now, let $F$ be some fixed forest of order at most $n$ from ${\cal F}_{\Delta}$.
Let $K_1,\ldots,K_\ell$ be the components of $F$.
For $i\in [\ell]$, let $n_i$ be the order of $K_i$ and 
let $F_{[i]}=K_1\cup\cdots\cup K_i$.
Since,
\begin{eqnarray*}
|{\cal T}(K_{i+1})| & \leq & (n_{i+1}+1)^p\leq (\Delta+1)^p,\\
|{\cal T}(F_{[i]})| & \leq & \left(1+\sum\limits_{j=1}^in_j\right)^p\leq (n+1)^p,\\
F_{[i+1]}&=&F_{[i]}\cup K_{i+1},\mbox{ and, hence, }\\
{\cal T}(F_{[i+1]}) &=& \left\{ t'+t'':t'\in {\cal T}(F_{[i]})\mbox{ and }t''\in {\cal T}(K_{i+1})\right\},
\end{eqnarray*}
a simple dynamic programming procedure allows to determine 
in time $O(n^k)$, 
for some $k$ depending only on $\Delta$,
the set ${\cal T}(F)$ and 
$${\rm lcs}(F_1,F_2)
=\max\left\{\sum\limits_{i=1}^pt_im(T_i):(t_1,\ldots,t_p)\in {\cal T}(F_1)\cap {\cal T}(F_2)\right\}.$$
Along the dynamic programming, 
one can also maintain suitable realizers
and the desired statement follows.
\end{proof}
As explained after Theorem \ref{theorem2},
we approximate the two given forests 
by simpler forests that are composed of copies of few
different trees that are structurally close to stars.
The following two lemmas contain the details.

Let $\epsilon>0$ and let $\Delta$ be a positive integer.

Let 
\begin{eqnarray}\label{e2}
\{ T_1,\ldots,T_p\}
\end{eqnarray}
be the set of all rooted trees of order at most $\Delta$,
where $T_p$ is the rooted tree of order $1$.
It is well-known 
that the number of rooted non-isomorphic trees of order $n+1$ 
is at most 
the {\it $n$-th Catalan number} $C_n=\frac{1}{n+1}{2n\choose n}$ \cite{st}.
Since the Catalan numbers are non-decreasing, 
$$p\leq \sum\limits_{i=1}^\Delta C_{i-1}<\Delta C_{\Delta}<{2\Delta\choose \Delta}.$$
Let 
\begin{eqnarray}\label{e3}
D(\epsilon,\Delta)=
[\Delta]_0\cup \left\{ \left\lceil(1+\epsilon)^i\right\rceil:i\in \mathbb{N}_0\right\}.
\end{eqnarray}
We say that a forest $F$ is {\it $(\epsilon,\Delta)$-clean} 
if each component $K$ of $F$ has a root vertex $r_K$ such that 
\begin{enumerate}[(i)]
\item every component of $K-r_K$ has order at most $\Delta$,
\item the degree $d_F(r_K)$ of $r_K$ in $F$ belongs to $D(\epsilon,\Delta)$, and
\item for every rooted tree $T$ in $\{ T_1,\ldots,T_{p-1}\}$,
that is, the order of $T$ is at least $2$,
the number of components $L$ of $K-r_K$, 
considered as trees rooted in the neighbor of $r_K$ in $V(L)$,
that are isomorphic to $T$ as a rooted tree is a multiple of
\begin{eqnarray}\label{e8b} 
\max\left\{1,\left\lfloor\frac{\epsilon d_F(r_K)}{\Delta{2\Delta\choose \Delta}}\right\rfloor\right\}.
\end{eqnarray}
\end{enumerate}
See Figure \ref{fig1} for an illustration.

\begin{figure}[H]
    \centering
    \begin{tikzpicture} [scale=0.3]
    \tikzstyle{point}=[draw,circle,inner sep=0.cm, minimum size=1mm, fill=black]
    \tikzstyle{point2}=[draw,circle,inner sep=0.cm, minimum size=0.5mm, fill=black]
    
    \draw [rounded corners, rounded corners=5mm] (-1.12,-2)--(5,8.6)--(11.12,-2)--cycle;    
    
    \node[point] (v) at (5,7) [label=below:] {};
    \node[point] (u) at (20,13) [label=above right:\text{$r_K$ root of degree in $D(\epsilon, \Delta)$}] {};
    \draw (v) --(u); 
    \node[align=center] at (5,2) {rooted\\subtree\\of order\\at most $\Delta$};

    \foreach \i in {-1,0,1}{
        \node[point2] at (0.5*\i+8,7) [label=left:] {};
    }

    \foreach \i in {0,1,2,3}{
        \node[point] (a) at (2*\i+11,7) [label=left:] {};
    \draw (a) -- (u);
    }
    \draw[thick,black,decorate,decoration={brace,amplitude=5}] (17.5,6) -- (10.5,6) node[midway, below,yshift=0]{};    
    \node[align=left] at (14,3) {4 copies \\ of $T_p=$};
    \node[point] at (16.5,2.3) [label=below:] {};

    \foreach \i in {-1,0,1}{
        \node[point2] at (0.5*\i+20,7) [label=left:] {};
    }

    \begin{scope}[shift={(12,0)}]   
    \foreach \i in {0,2}{
        \node[point] (a) at (2*\i+11,7) [label=left:] {};
        \node[point] (b) at (2*\i+11+2*0.5,7-2*0.866) [label=left:] {};
        \node[point] (c) at (2*\i+11-2*0.5,7-2*0.866) [label=left:] {};
        \node[point] (d) at (2*\i+11+4*0.5,7-4*0.866) [label=left:] {};
        \node[point] (e) at (2*\i+11,7-4*0.866) [label=left:] {};
        \draw (a) -- (d);
        \draw (a) -- (c);
        \draw (e) -- (b);
      \draw (a) -- (u);
    }
    \draw[thick,black,decorate,decoration={brace,amplitude=5}] (17.5,7-4*0.866-1) -- (10.5,7-4*0.866-1) node[midway, below,yshift=0]{};    
    \node[align=left] at (14,7-4*0.866-4) {2 copies \\ of};
    
    \begin{scope}[shift={(3.5,-4*0.866-5)}]   
        \foreach \i in {0}{
        \node[point] (a) at (2*\i+11,7) [label=left:] {};
        \node[point] (b) at (2*\i+11+2*0.5,7-2*0.866) [label=left:] {};
        \node[point] (c) at (2*\i+11-2*0.5,7-2*0.866) [label=left:] {};
        \node[point] (d) at (2*\i+11+4*0.5,7-4*0.866) [label=left:] {};
        \node[point] (e) at (2*\i+11,7-4*0.866) [label=left:] {};
        \draw (a) -- (d);
        \draw (a) -- (c);
        \draw (e) -- (b);
        }
    \end{scope}

    \foreach \i in {-1,0,1}{
        \node[point2] at (0.5*\i+18,7) [label=left:] {};
    }

    \end{scope}

    \begin{scope}[shift={(23,0)}]   
        \foreach \i in {0,2,4}{
        \node[point] (a) at (2*\i+11,7) [label=left:] {};
        \node[point] (b) at (2*\i+11,7-2*0.866) [label=left:] {};
        \node[point] (c) at (2*\i+11,7-4*0.866) [label=left:] {};
        \node[point] (d) at (2*\i+11-1,7-6*0.866) [label=left:] {};
        \node[point] (e) at (2*\i+11+1,7-6*0.866) [label=left:] {};
        \draw (a) -- (c);
        \draw (d) -- (c) -- (e);
        \draw (a) -- (u);
    }
    \draw[thick,black,decorate,decoration={brace,amplitude=5}] (20.5,7-6*0.866-1) -- (9.5,7-6*0.866-1) node[midway, below,yshift=0]{};    
    \node[align=left] at (15,7-6*0.866-4) {3 copies \\ of};
    \begin{scope}[shift={(4.5,-6*0.866-5)}]   
        \foreach \i in {0}{
        \node[point] (a) at (2*\i+11,7) [label=left:] {};
        \node[point] (b) at (2*\i+11,7-2*0.866) [label=left:] {};
        \node[point] (c) at (2*\i+11,7-4*0.866) [label=left:] {};
        \node[point] (d) at (2*\i+11-1,7-6*0.866) [label=left:] {};
        \node[point] (e) at (2*\i+11+1,7-6*0.866) [label=left:] {};
        \draw (a) -- (c);
        \draw (d) -- (c) -- (e);
        }
    \end{scope}

    \end{scope}

\end{tikzpicture}

\caption{A component $K$ with root $r_K$ of an $ (\epsilon,\Delta)$-clean forest $F$.
Note that we consider the components $L$ of $K-r_K$ as trees rooted 
in the neighbor of $r_K$ in $V(L)$, which means that we distinguish
isomorphic components of $K-r_K$ 
that are attached differently to the root $r_K$.
The figure shows five isomorphic components of $K-r_K$,
two of which are attached to $r_k$ at their unique vertex of degree $2$
and three are attached to $r_k$ differently.}\label{fig1}
\end{figure}
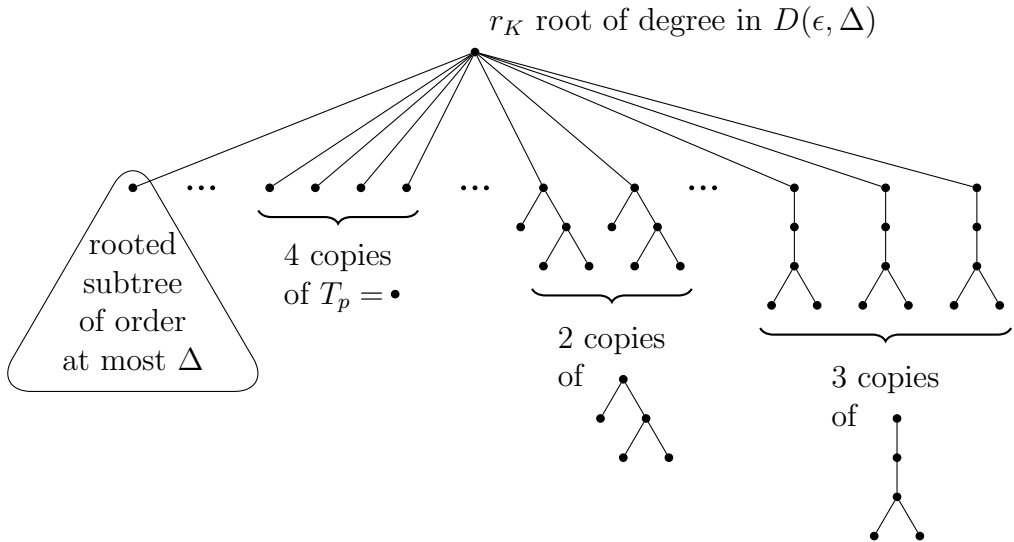

\begin{lemma}\label{lemma2}
For every $\epsilon\in (0,1)$ and $\Delta\in \mathbb{N}$ with $\epsilon \Delta\geq 1$,
there is some $k\in \mathbb{N}$ with the following property:
For a given forest $F$ of order $n$,
one can determine a spanning $(\epsilon,\Delta)$-clean subforest $F'$ of $F$ with 
$$m(F')\geq \left(1-4\left(\epsilon+\frac{1}{\Delta}\right)\right)m(F)$$
in time $O(n^k)$.
\end{lemma}
\begin{proof}
Let $F$ be a forest of size $m$.
Since the degree sum of $F$ is $2m$,
there are less than $\frac{2m}{\Delta}$ vertices in $F$ 
that have degree more than $\Delta$.
Rooting every component of $F$ in some vertex
and choosing, for every vertex of degree more than $\Delta$ that is no root,
the edge to its parent,
yields a set $E_0$ of less than $\frac{2m}{\Delta}$ edges of $F$
such that every component of $F_0=F-E_0$ 
contains at most one vertex of degree more than $\Delta$.
We are now going to ensure 
the three properties (i), (ii), and (iii) from the definition of $(\epsilon,\Delta)$-cleanness
by removing further edges in three consecutive cleaning steps.

For every component $K$ of $F_0$, 
choose a vertex $r_K$ of maximum degree within $K$ as its root. 
Call a component $K$ {\it $\frac{1}{3}$-clean} 
if every component of $K-r_K$ has order at most $\Delta$.
Suppose that $K$ is a component that is not $\frac{1}{3}$-clean.
Let $u$ be some vertex of maximum depth in $K$ rooted in $r_K$
such that one plus the number of descendants of $u$ in $K$ is more than $\Delta$.
Removing the edge between $u$ and its parent
cuts off a component $L$ containing $u$ that has at least $\Delta$ edges.
Choosing $u$ as its root $r_L$, the component $L$ is $\frac{1}{3}$-clean.
For the remaining part $K-V(L)$ of $K$,
we keep $r_K$ as its root.
Iteratively repeating this procedure as long as there are 
components that are not $\frac{1}{3}$-clean,
yields a set $E_1$ 
of at most $\frac{m(F_0)}{\Delta}\leq \frac{m}{\Delta}$ edges of $F_0$
such that every component $K$ of $F_1=F_0-E_1$ 
has a specified root $r_K$ and is $\frac{1}{3}$-clean.

For a component $K$ of $F_1$, let $d_K$ denote the degree of $r_K$ in $F_1$.
Call a component $K$ {\it $\frac{2}{3}$-clean} 
if it is $\frac{1}{3}$-clean and $d_K$ belongs to $D(\epsilon,\Delta)$.
For $i\in\mathbb{N}_0$, let $d_i=\left\lceil(1+\epsilon)^i\right\rceil$,
that is, 
$D(\epsilon,\Delta)=[\Delta]_0\cup \left\{ d_i:i\in \mathbb{N}_0\right\}$.
For every component $K$ of $F_1$ with $d_K\not\in D(\epsilon,\Delta)$, 
let $i_K$ be the largest non-negative integer $i$ with $d_i\leq d_K$,
in particular,
$$(1+\epsilon)^{i_K}
\leq  d_{i_K}
\leq d_K
<d_{i_K+1}
<(1+\epsilon)^{i_K+1}+1.$$
Since $d_K\not\in D(\epsilon,\Delta)$, we have $d_K-1\geq \Delta$.
This implies 
$1\leq \epsilon \Delta\leq \epsilon(d_K-1)\leq \epsilon(1+\epsilon)^{i_K+1}$,
and, hence,
$$d_K
<(1+\epsilon)^{i_K+1}+1
\leq (1+\epsilon)^{i_K+1}+\epsilon(1+\epsilon)^{i_K+1}
=(1+\epsilon)^{i_K+2}.$$
Let the set $E_2$ of edges of $F_1$ contain exactly 
$d_K-d_{i_K}$ many edges incident with $r_K$
for every component $K$ of $F_1$ with $d_K\not\in D(\epsilon,\Delta)$.
Since
$$\frac{d_K-d_{i_K}}{d_K}
\leq \frac{(1+\epsilon)^{i_K+2}-(1+\epsilon)^{i_K}}{(1+\epsilon)^{i_K}}
=2\epsilon+\epsilon^2
\leq 3\epsilon,$$
the set $E_2$ contains at most a $3\epsilon$-fraction of the edges of $F_1$,
that is, $|E_2|\leq 3\epsilon m(F_1)\leq 3\epsilon m$.
Let $F_2=F_1-E_2$.
For every component $K$ of $F_2$ 
that contains the root $r$ of some component of $F_1$,
choose $r$ as the root $r_K$ of $K$.
Each component $K$ of $F_2$ 
that does not contain the root of some component of $F_1$
has order at most $\Delta$, and 
we choose any of its vertices as its root $r_K$.
With these choices of the roots, each component of $F_2$ is $\frac{2}{3}$-clean.

Call a component $K$ of $F_2$ {\it clean} 
if it is $\frac{2}{3}$-clean and, 
for every rooted tree $T$ of order at least $2$,
the number of components $L$ of $K-r_K$, 
considered as trees rooted in the neighbor of $r_K$ in $V(L)$,
that are isomorphic to $T$ as a rooted tree 
is a multiple of 
\begin{eqnarray*}\label{e8}
\delta=\max\left\{ 1,\left\lfloor\frac{\epsilon d_{F_2}(r_K)}{\Delta{2\Delta\choose \Delta}}\right\rfloor\right\}.
\end{eqnarray*}
Suppose that some component $K$ of $F_2$ is not clean.
This implies that 
$\delta=\left\lfloor\frac{\epsilon d_{F_2}(r_K)}{\Delta{2\Delta\choose \Delta}}\right\rfloor>1$.
For every component $L$ of $K-r_K$, 
choose the neighbor of $r_K$ in $V(L)$ as its root.
Let $\{ T_1,\ldots,T_p\}$ be as in (\ref{e2}).
Let $K-r_K$ contain exactly $t_i$ components that are isomorphic to $T_i$ 
as rooted trees for every $i\in [p]$.
Now, for every $i\in [p-1]$, remove all edges of exactly 
$t_i-\left\lfloor\frac{t_i}{\delta}\right\rfloor \delta$
copies of the rooted tree $T_i$ among the components of $K-r_K$.
This does not affect the degree of $r_K$ and 
results in a subforest $K'$ of $K$ 
in which the component containing the root $r_K$ is clean
and all remaining components are isolated vertices,
which means that they are also clean.
Furthermore, since each $T_i$ has at most $\Delta-1$ edges, we obtain
\begin{eqnarray*}
|E(K)\setminus E(K')|&=&m(K)-m(K') 
\leq (p-1)(\Delta-1)\delta
< {2\Delta\choose \Delta}\Delta\delta
\leq \epsilon d_{F_2}(r_K)
\leq \epsilon m(K).
\end{eqnarray*}
Performing this operation for every component of $F_2$ that is not clean,
yields a forest $F'=F_2-E_3$ that is $(\epsilon,\Delta)$-clean.
The set $E_3$ of removed edges satisfies 
$|E_3|\leq \epsilon m(F_2)\leq \epsilon m$.

Now,
$$m(F')=m-|E_0|-|E_1|-|E_2|-|E_3|
\geq m-\left(\frac{2}{\Delta}+\frac{1}{\Delta}+3\epsilon+\epsilon\right)m
\geq m-4\left(\epsilon+\frac{1}{\Delta}\right)m.$$
Furthermore, all steps of the cleaning procedures 
can be performed in polynomial time 
for fixed $\epsilon$ and $\Delta$.
This completes the proof.
\end{proof}
A natural choice for $\Delta$ is $\left\lceil\frac{1}{\epsilon}\right\rceil$,
which immediately implies $\epsilon \Delta\geq 1$.

Accordingly, 
let a forest be {\it $\epsilon$-clean} if it is 
$\left(\epsilon,\left\lceil\frac{1}{\epsilon}\right\rceil\right)$-clean and
let 
\begin{eqnarray}
\label{e4}
{\cal T}(\epsilon) & = & \{ T_1,\ldots,T_p\}
\mbox{ be as in (\ref{e2}) for $\Delta=\left\lceil\frac{1}{\epsilon}\right\rceil$ as well as}\\
\label{e5}
D(\epsilon) & = & D(\epsilon,\Delta) 
\mbox{ be as in (\ref{e3}) for $\Delta=\left\lceil\frac{1}{\epsilon}\right\rceil$.}
\end{eqnarray}

\begin{lemma}\label{lemma4}
For every $\epsilon\in (0,1)$, 
there are $c_1,c_2,k\in \mathbb{N}$ 
with the following property:
For every positive integer $n$ at least $2$, 
there is a set ${\cal S}(\epsilon,n)$ 
of at most $c_1\log (n)$ rooted trees 
such that every component 
of every $\epsilon$-clean forest of order $n$
belongs to ${\cal S}(\epsilon,n)$.
Furthermore, 
if $S_1,\ldots,S_q$ is a linear ordering of the elements of ${\cal S}(\epsilon,n)$
such that the degrees $d_i$ of the roots $r_i$ of $S_i$ 
are non-decreasing along this ordering,
and $i,j\in [q]$ are such that $j\geq i+c_2$, 
then $d_i\leq \epsilon d_j$.
Finally, ${\cal S}(\epsilon,n)$ can be constructed in time $O(n^k)$.
\end{lemma}
\begin{proof}
Let $\Delta=\left\lceil\frac{1}{\epsilon}\right\rceil$.
Let $F$ be an $\epsilon$-clean forest of order $n$.
Let $K$ be a component of $F$ with root vertex $r_K$.

If $d_F(r_K)<\frac{\Delta}{\epsilon}{2\Delta\choose \Delta}$, then 
$$n(K)\leq 1+d_F(r_K)\Delta\leq 1+\frac{\Delta^2}{\epsilon}{2\Delta\choose \Delta},$$
which implies that, for fixed $\epsilon$, 
there are finitely many choices for such a component.

Now, let $d_F(r_K)\geq \frac{\Delta}{\epsilon}{2\Delta\choose \Delta}$.
For $\delta$ as in (\ref{e8b}), we obtain
$$\delta
=\max\left\{ 1,\left\lfloor\frac{\epsilon d_F(r_K)}{\Delta{2\Delta\choose \Delta}}\right\rfloor\right\}
=\left\lfloor\frac{\epsilon d_F(r_K)}{\Delta{2\Delta\choose \Delta}}\right\rfloor
\geq \frac{\epsilon d_F(r_K)}{2\Delta{2\Delta\choose \Delta}}.$$
As in (\ref{e4}),
let $\{ T_1,\ldots,T_p\}$ be the set of all rooted trees of order at most $\Delta$,
where $T_p$ is the rooted tree of order $1$.
Let $K-r_K$ contain exactly $t_i$ components 
that are isomorphic to $T_i$ as rooted trees for every $i\in [p]$.
Since $t_i\leq d_F(r_K)$ for every $i\in [p-1]$,
property (iii) in the definition of $(\epsilon,\Delta)$-cleanness implies
that there are at most
$$1+\frac{d_F(r_K)}{\delta}\leq 1+\frac{2\Delta{2\Delta\choose \Delta}}{\epsilon}$$
possible values for each $t_i$ with $i\in [p-1]$.
Recall that $T_p$ is the rooted tree of order $1$ and that 
$$t_p=d_F(r_K)-(t_1+\cdots+t_{p-1}),$$ 
which implies that $K$ is determined up to isomorphism 
by $t_1,\ldots,t_{p-1}$ and the degree of its root. 
Therefore, for every fixed integer $d$ with 
$\frac{\Delta}{\epsilon}{2\Delta\choose \Delta}\leq d\leq n$,
there are at most
$c_3:=\left(1+\frac{2\Delta{2\Delta\choose \Delta}}{\epsilon}\right)^{p-1}$
many choices for $K$ such that the degree $d_F(r_K)$ of its root $r_K$ equals $d$.
Recall that $p$ is bounded in terms of $\Delta$, which,
in turn, is bounded in terms of $\epsilon$.
Since $D(\epsilon)$ as in (\ref{e5}) contains at most 
$\log_{(1+\epsilon)}(n)=\frac{\log(n)}{\log(1+\epsilon)}$
such values $d$, 
there is some integer $c_1$ depending only on $\epsilon$, and 
there is a set ${\cal S}(\epsilon,n)$ of at most $c_1\log (n)$ rooted trees 
such that every component of every $\epsilon$-clean forest of order $n$
belongs to ${\cal S}(\epsilon,n)$.

Now, let 
$$S_1,\ldots,S_q$$ 
be a linear ordering 
of the elements of ${\cal S}(\epsilon,n)$
such that the degrees $d_i$ of the roots $r_i$ of $S_i$ 
are non-decreasing along this ordering.
This ordering begins with finitely many rooted trees 
with roots of degrees $d_i$ at most $\frac{\Delta}{\epsilon}{2\Delta\choose \Delta}$.
Once the root degrees $d_i$ are at least this value,
the structure of $D(\epsilon)$ implies that 
%--- up to the asymptotically negligible rounding effect of the powers of $(1+\epsilon)$ ---
they increase by a factor of $(1+\epsilon)$
after every $O(c_3)$ steps in the ordering.
This implies the existence of some positive integer $c_2$ such that,
for every $i,j\in [q]$ with $j\geq i+c_2$, 
we have $d_i\leq \epsilon d_j$.

The above arguments imply that 
${\cal S}(\epsilon,n)$ can be constructed in time $O(n^k)$
for some integer $k$ depending only on $\epsilon$.
\end{proof}
We are now in a position to complete the proof.
\begin{proof}[Proof of Theorem \ref{theorem2}]
Let $\epsilon\in (0,1)$ be fixed.
Let $\Delta=\left\lceil\frac{1}{\epsilon}\right\rceil$.
Within this proof we call a forest {\it clean} if it is 
$\left(\epsilon,\left\lceil\frac{1}{\epsilon}\right\rceil\right)$-clean.
Let ${\cal T}={\cal T}(\epsilon)$ be as in (\ref{e5}), that is,
$${\cal T}=\{ T_1,\ldots,T_p\}$$
is the set of all rooted trees of order at most $\Delta$,
where $T_p$ is the rooted tree of order $1$.
Let ${\cal D}={\cal D}(\epsilon)$ be as in (\ref{e5}), that is,
${\cal D}
=[\Delta]_0\cup \left\{ \left\lceil(1+\epsilon)^i\right\rceil:i\in \mathbb{N}_0\right\}.$

Now, let $F_1$ and $F_2$ be two given forests of order $n$ at least $2$,
for which we want to determine a common subgraph $F$ of large size.
Note that, in view of the desired statement, 
it would suffice that $m(F)\geq {\rm lcs}(F_1,F_2)-C\epsilon n$ 
for some constant $C$ independent of $\epsilon$ and $n$.

\medskip

\medskip

\medskip

\noindent {\bf Cleaning the given forests}

Suppose that $F_1$ or $F_2$ are not clean.
Using Lemma \ref{lemma2},
we can determine in polynomial time
a set $E_1$ of edges of $F_1$ and a set $E_2$ of edges of $F_2$
such that $F_1'=F_1-E_1$ and $F_2'=F_2-E_2$ are clean and 
$$|E_1|+|E_2|
\leq 4\left(\epsilon+\frac{1}{\Delta}\right)(m(F_1)+m(F_2))
\leq 4\left(\epsilon+\epsilon\right)2n
=16\epsilon n.$$
If $F$ is a common subgraph of $F_1$ and $F_2$,
then removing from $F$ the at most $16\epsilon n$ edges corresponding 
to edges from $E_1$ or $E_2$ that belong to $F$
yields a common subgraph $F'$ of $F_1'$ and $F_2'$
such that $m(F')\geq m(F)-16\epsilon n$,
in particular, ${\rm lcs}(F'_1,F'_2)\geq {\rm lcs}(F_1,F_2)-16\epsilon n$.
In view of the desired statement, 
we may therefore assume that 
$$\mbox{$F_1$ and $F_2$ are clean.}$$
Using Lemma \ref{lemma4}, we construct in polynomial time the set
$${\cal S}={\cal S}(\epsilon,n)=\{ S_1,\ldots,S_q\}$$
and the integers $c_1$ and $c_2$ as in Lemma \ref{lemma4}, that is, 
${\cal S}$ contains $q\leq c_1\log(n)$ clean rooted trees and
every component of $F_1$ and $F_2$ belongs to ${\cal S}$.
Furthermore, denoting the root of $S_i$ and its degree by $r_i$ and $d_i$, 
respectively, we have 
\begin{eqnarray}\label{e7}
\mbox{$d_i\leq \epsilon d_j$ for every $i,j\in [q]$ with $j\geq i+c_2$.}
\end{eqnarray}

\medskip

\medskip

\pagebreak 

\noindent {\bf Notational interlude}

Let $F$ be a common subgraph of $F_1$ and $F_2$.
Extending an isomorphism between 
a subgraph of $F_1$ that is isomorphic to $F$
and 
a subgraph of $F_2$ that is also isomorphic to $F$
yields a bijection $f:V(F_1)\to V(F_2)$ 
with the following property:
$F$ is isomorphic to a subgraph of the forest $F_f$ with vertex set $V(F_1)$ 
that contains all edges $uv$ of $F_1$ for which $f(u)f(v)$ is an edge in $F_2$.
In fact, $F_f$ itself is a spanning common subgraph of $F_1$ and $F_2$, and
$${\rm lcs}(F_1,F_2)=\max\{ m(F_f):f:V(F_1)\to V(F_2)\mbox{ bijective}\}.$$
Possibly after adding isolated vertices and renaming vertices,
we may assume now and later, for notational convenience,
that $F$ is a spanning subgraph of $F_f$.
See Figure \ref{fig3} for an illustration.

\begin{figure}[H]
    \centering
    \begin{tikzpicture} [scale=0.25]
    \tikzstyle{point}=[draw,circle,inner sep=0.cm, minimum size=1mm, fill=black]
    \tikzstyle{point2}=[draw,circle,inner sep=0.cm, minimum size=0.5mm, fill=black]
    \tikzstyle{line1}=[line width=0.4mm]
    \tikzstyle{line2}=[thin]

\node[point] (a) at (0,0) [label=above:] {};
\coordinate (a) at (0,0) [label=above:] {};
\node[point] (b) at (0,2) [label=above:] {};
\coordinate (b) at (0,2) [label=above:] {};
\draw[line1] (a) -- (b);
\node[point] (c) at (0,4) [label=above:] {};
\coordinate (c) at (0,4) [label=above:] {};
\node[point] (d) at (0,6) [label=above:] {};
\coordinate (d) at (0,6) [label=above:] {};
\draw[line1] (c) -- (d);
\node[point] at (0,10) [label=above:] {};
\node[point] (d) at (0,14) [label=above:] {};
\coordinate (d) at (0,14) [label=above:] {};
\node[point] (e) at (2,12) [label=above:] {};
\coordinate (e) at (2,12) [label=above:] {};
\draw[line1] (d) -- (e);
\node[point] at (4,14) [label=above:] {};
\node[point] (x) at (6,12) [label=above:] {};
\coordinate (x) at (6,12) [label=above:] {};
\node[point] (y) at (6,8) [label=above:] {};
\coordinate (y) at (6,8) [label=above:] {};
\node[point] at (6,10) [label=above:] {};
\draw[line1] (x) -- (y);
\node[point] at (8,10) [label=above:] {};
\node[point] (a) at (4,2) [label=above:] {};
\coordinate (a) at (4,2) [label=above:] {};
\node[point] (c) at (4,4) [label=above:] {};
\coordinate (c) at (4,4) [label=above:] {};
\node[point] (b) at (6,4) [label=above:] {};
\coordinate (b) at (6,4) [label=above:] {};
\draw[line1] (a) -- (b) -- (c);
\node[point] at (6,2) [label=above:] {};
\node[point] (a) at (6,0) [label=above:] {};
\coordinate (a) at (6,0) [label=above:] {};
\node[point] (b) at (8,0) [label=above:] {};
\coordinate (b) at (8,0) [label=above:] {};
\node[point] (c) at (8,2) [label=above:] {};
\coordinate (c) at (8,2) [label=above:] {};
\draw[line1] (a) -- (b) -- (c);
\node at (4,18) [label=above:] {$F$};

\node at (12,7) [label=above:] {$\subseteq $};

\begin{scope}[shift={(2,0)}]
\begin{scope}[shift={(14,0)}]

\node[point] (a) at (0,0) [label=above:] {};
\coordinate (a) at (0,0) [label=above:] {};
\node[point] (b) at (0,2) [label=above:] {};
\coordinate (b) at (0,2) [label=above:] {};
\draw[line1] (a) -- (b);
\node[point] (c) at (0,4) [label=right:$y$] {};
\coordinate (c) at (0,4) [label=above:] {};
\draw[line2] (b) -- (c);
\node[point] (d) at (0,6) [label=above:] {};
\coordinate (d) at (0,6) [label=above:] {};
\draw[line1] (c) -- (d);
\node[point] (x) at (0,10) [label=above:] {};
\coordinate (x) at (0,10) [label=above:] {};
\node[point] (d) at (0,14) [label=above:] {};
\coordinate (d) at (0,14) [label=above:] {};
\node[point] (e) at (2,12) [label=right:$u$] {};
\coordinate (e) at (2,12) [label=above:] {};
\draw[line1] (d) -- (e);
\node[point] (y) at (4,14) [label=right:] {};
\node at (3.5,14) [label=right:$f^{-1}(v)$] {};

\coordinate (y) at (4,14) [label=above:] {};
\draw[line2] (x) -- (y);
\node[point] (x) at (6,12) [label=right:$x$] {};
\coordinate (x) at (6,12) [label=above:] {};
\node[point] (y) at (6,8) [label=above:] {};
\coordinate (y) at (6,8) [label=above:] {};
\node[point] at (6,10) [label=above:] {};
\draw[line1] (x) -- (y);
\node[point] (z) at (8,10) [label=above:] {};
\coordinate (z) at (8,10) [label=above:] {};
\draw[line2] (y) -- (z);
\node[point] (a) at (4,2) [label=above:] {};
\coordinate (a) at (4,2) [label=above:] {};
\node[point] (c) at (4,4) [label=above:] {};
\coordinate (c) at (4,4) [label=above:] {};
\node[point] (b) at (6,4) [label=above:] {};
\coordinate (b) at (6,4) [label=above:] {};
\draw[line1] (a) -- (b) -- (c);
\node[point] (d) at (6,2) [label=above:] {};
\coordinate (d) at (6,2) [label=above:] {};
\node[point] (a1) at (6,0) [label=above:] {};
\coordinate (a1) at (6,0) [label=above:] {};
\node[point] (b1) at (8,0) [label=above:] {};
\coordinate (b1) at (8,0) [label=above:] {};
\node[point] (c1) at (8,2) [label=above:] {};
\coordinate (c1) at (8,2) [label=above:] {};
\draw[line1] (a1) -- (b1) -- (c1);
\draw[line2] (a) -- (d) -- (a1);
\node at (4,18) [label=above:] {$F_1$};

\begin{scope}[shift={(2,0)}]    
\draw[line1,->] (10,10) -- (14,10);
\node at (12,11) [label=above:] {$f$};
\draw[line1,<-] (10,4) -- (14,4);
\node at (12,5) [label=above:] {$f^{-1}$};
\end{scope}

\end{scope}

\begin{scope}[shift={(6,0)}]
\begin{scope}[shift={(28,0)}]

\node[point] (a) at (0,0) [label=below:$f(x)$] {};
\coordinate (a) at (0,0) [label=above:] {};
\node[point] (b) at (0,2) [label=above:] {};
\coordinate (b) at (0,2) [label=above:] {};
\node[point] (c) at (0,4) [label=above:] {};
\coordinate (c) at (0,4) [label=above:] {};
\draw[line1] (a) -- (c);
\node[point] (d) at (0,6) [label=above:] {};
\coordinate (d) at (0,6) [label=above:] {};
\node[point] (d) at (2,2) [label=above:] {};
\coordinate (d) at (2,2) [label=right:] {};
\draw[line2] (b) -- (d);
\node[point] (x) at (0,10) [label=above:] {};
\coordinate (x) at (0,10) [label=above:] {};
\node at (1.5,7.5) [label=right:$f(u)$] {};
\node[point] (y) at (2,8) [label=right:] {};
\coordinate (y) at (2,8) [label=above:] {};
\draw[line1] (x) -- (y);
\node[point] (z) at (4,10) [label=left:$v$] {};
\coordinate (z) at (4,10) [label=right:] {};
\draw[line2] (c) -- (y) -- (z);
\node[point] (a) at (4,12) [label=above:] {};
\coordinate (a) at (4,12) [label=above:] {};
\node[point] (b) at (6,14) [label=above:] {};
\coordinate (b) at (6,14) [label=above:] {};
\node[point] (c) at (4,14) [label=above:] {};
\coordinate (c) at (4,14) [label=above:] {};
\draw[line1] (a) -- (b) -- (c);
\node[point] (d) at (6,12) [label=above:] {};
\coordinate (d) at (6,12) [label=above:] {};
\node[point] (e) at (8,12) [label=above:] {};
\coordinate (e) at (8,12) [label=above:] {};
\draw[line2] (d) -- (b) -- (e);
\node[point] (f) at (6,10) [label=above:] {};
\coordinate (f) at (6,10) [label=above:] {};
\node[point] (g) at (8,10) [label=above:] {};
\coordinate (g) at (8,10) [label=above:] {};
\draw[line1] (f) --(g)--(e);
\node[point] (a) at (8,0) [label=above:] {};
\coordinate (a) at (8,0) [label=above:] {};
\node[point] (b) at (8,2) [label=above:] {};
\coordinate (b) at (8,2) [label=above:] {};
\draw[line1] (a)--(b);
\node[point] (c) at (8,4) [label=right:$f(y)$] {};
\coordinate (c) at (8,4) [label=above:] {};
\node[point] (d) at (8,6) [label=above:] {};
\coordinate (d) at (8,6) [label=above:] {};
\draw[line1] (c) -- (d);
\draw[line2] (b) arc(-90:-270:2);
\node at (4,18) [label=above:] {$F_2$};
\node at (12,7) [label=above:] {$\supseteq $};

\end{scope}

\begin{scope}[shift={(44,0)}]

\node[point] (a) at (0,0) [label=above:] {};
\coordinate (a) at (0,0) [label=above:] {};
\node[point] (b) at (0,2) [label=above:] {};
\coordinate (b) at (0,2) [label=above:] {};
\node[point] (c) at (0,4) [label=above:] {};
\coordinate (c) at (0,4) [label=above:] {};
\draw[line1] (a) -- (c);
\node[point] (d) at (0,6) [label=above:] {};
\coordinate (d) at (0,6) [label=above:] {};
\node[point] (d) at (2,2) [label=above:] {};
\coordinate (d) at (2,2) [label=above:] {};
\node[point] (x) at (0,10) [label=above:] {};
\coordinate (x) at (0,10) [label=above:] {};
\node[point] (y) at (2,8) [label=above:] {};
\coordinate (y) at (2,8) [label=above:] {};
\draw[line1] (x) -- (y);
\node[point] (z) at (4,10) [label=above:] {};
\coordinate (z) at (4,10) [label=above:] {};
\node[point] (a) at (4,12) [label=above:] {};
\coordinate (a) at (4,12) [label=above:] {};
\node[point] (b) at (6,14) [label=above:] {};
\coordinate (b) at (6,14) [label=above:] {};
\node[point] (c) at (4,14) [label=above:] {};
\coordinate (c) at (4,14) [label=above:] {};
\draw[line1] (a) -- (b) -- (c);
\node[point] (d) at (6,12) [label=above:] {};
\coordinate (d) at (6,12) [label=above:] {};
\node[point] (e) at (8,12) [label=above:] {};
\coordinate (e) at (8,12) [label=above:] {};
\node[point] (f) at (6,10) [label=above:] {};
\coordinate (f) at (6,10) [label=above:] {};
\node[point] (g) at (8,10) [label=above:] {};
\coordinate (g) at (8,10) [label=above:] {};
\draw[line1] (f) --(g)--(e);
\node[point] (a) at (8,0) [label=above:] {};
\coordinate (a) at (8,0) [label=above:] {};
\node[point] (b) at (8,2) [label=above:] {};
\coordinate (b) at (8,2) [label=above:] {};
\draw[line1] (a)--(b);
\node[point] (c) at (8,4) [label=above:] {};
\coordinate (c) at (8,4) [label=above:] {};
\node[point] (d) at (8,6) [label=above:] {};
\coordinate (d) at (8,6) [label=above:] {};
\draw[line1] (c) -- (d);
\node at (4,18) [label=above:] {$f(F)$};
\end{scope}
\end{scope}
\end{scope}

\end{tikzpicture}
\caption{A common subgraph $F$ of $F_1$ and $F_2$
as a spanning subgraph of $F_1$ 
together with a corresponding bijection $f:V(F_1)\to V(F_2)$.
The forest $f(F)$ with vertex set $V(F_2)$ and edge set
$\{ f(x)f(y):xy\in E(F)\}$ is a spanning subgraph of $F_2$.
Note that $F_f$ contains strictly more edges than $F$;
the edge $uf^{-1}(v)$ of $F_1$ belongs to $F_f$,
because the edge $f(u)v$ belongs to $F_2$.
Up to isomorphism, $F$ is described by the multiplicities of its components;
it consists of $4$ copies of $T_p$,
$3$ copies of $K_2$, and
$3$ copies of $P_3$.
Note that there are two non-isomorphic ways to choose a root for $P_3$.}\label{fig3}
\end{figure}
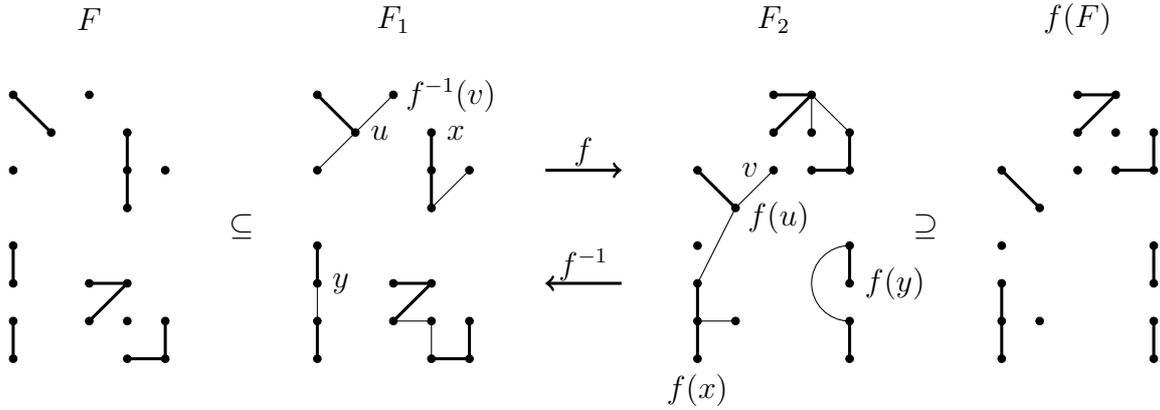

\medskip

\medskip

\noindent {\bf Nice solutions --- pairing or isolating large degree roots}

We call a common subgraph $F$ of $F_1$ and $F_2$ {\it nice}
if there is some bijection $f:V(F_1)\to V(F_2)$ 
such that $F$ is a spanning subgaph of $F_f$ 
and $F$ is a common subgraph 
of the two forests $F_1'$ and $F_2'$ 
constructed as follows:
\begin{itemize}
\item $F_1'$ is the spanning subgraph of $F_1$ 
that is obtained by removing all edges 
incident with every root vertex $r$ of some component of $F_1$ 
such that $d_{F_1}(r)\geq \frac{\Delta}{\epsilon}$ and
either 
$\epsilon d_{F_1}(r)\geq d_{F_2}(f(r))$
or 
$\epsilon d_{F_2}(f(r))\geq d_{F_1}(r)$.
Note that, in both cases, 
\begin{eqnarray}\label{e10}
\min\{ d_{F_1}(r),d_{F_2}(f(r))\}
\leq \epsilon \max\{ d_{F_1}(r),d_{F_2}(f(r))\}
\leq \epsilon (d_{F_1}(r)+d_{F_2}(f(r))).
\end{eqnarray}
\item $F_2'$ is the spanning subgraph of $F_2$ 
that is obtained by removing all edges 
incident with every root vertex $r$ of some component of $F_2$ 
such that $d_{F_2}(r)\geq \frac{\Delta}{\epsilon}$ and
either 
$\epsilon d_{F_2}(r)\geq d_{F_1}(f^{-1}(r))$
or 
$\epsilon d_{F_1}(f^{-1}(r))\geq d_{F_2}(r)$.
Again, in both cases, 
\begin{eqnarray*}
\min\{ d_{F_1}(f^{-1}(r)),d_{F_2}(r)\}
\leq \epsilon \max\{ d_{F_1}(f^{-1}(r)),d_{F_2}(r)\}
\leq \epsilon (d_{F_1}(f^{-1}(r))+d_{F_2}(r)).
\end{eqnarray*}
\end{itemize}
Note that if a vertex $r$ of $F_1$ of degree at least $\frac{\Delta}{\epsilon}$,
which is necessarily the root of some component of $F_1$,
is such that $f(r)$ is no root of some component of $F_2$,
then the degree of $f(r)$ in $F_2$ is at most $\Delta$.
It follows that $\epsilon d_{F_1}(r)\geq d_{F_2}(f(r))$, 
which implies that $r$ will be isolated in the nice subgraph $F$.

Now, suppose that $F$ is a common subgraph of $F_1$ and $F_2$
with $m(F)\geq {\rm lcs}(F_1,F_2)-\epsilon n$ that is not nice.
For convenience, we may assume that 
$F$ is a spanning subgraph of $F_f$ for some bijection $f$.
Since $F$ is not nice, 
$F$ is not a common subgraph of $F_1'$ and $F_2'$ 
as defined above using this $f$.
Since the degree of a root $r$ of some component of $F_1$ within the forest $F$ 
is at most $\min\{ d_{F_1}(r),d_{F_2}(f(r))\}$,
(\ref{e10}) implies that removing at most
\begin{eqnarray*} 
\sum\limits_{r\in V(F_1)}\min\{ d_{F_1}(r),d_{F_2}(f(r))\}
&\leq &\epsilon \sum\limits_{r\in V(F_1)} (d_{F_1}(r)+d_{F_2}(f(r)))\\
&=& \epsilon \left(\sum\limits_{r\in V(F_1)} d_{F_1}(r)
+\sum\limits_{r'\in V(F_2)}d_{F_2}(r')\right)\\
&=& 2\epsilon(m(F_1)+m(F_2))
\leq 4\epsilon n
\end{eqnarray*} 
edges from $F$ yields a subgraph $F'$ of $F_1'$.
Symmetrically, removing at most $4\epsilon n$ further edges from this subgraph $F'$
yields a subgraph $F''$ of $F_1'$ and $F_2'$.
Now, $F''$ is a nice common subgraph of $F_1$ and $F_2$, and 
$$m(F')\geq m(F)-8\epsilon n\geq {\rm lcs}(F_1,F_2)-9\epsilon n.$$
Therefore, in view of the desired statement,
it suffices to determine in polynomial time 
a nice common subgraph $F$ of $F_1$ and $F_2$
that is a spanning subgraph of $F_1$ and has many edges.

\medskip

\medskip

\medskip

\noindent {\bf (Potentially) large components in nice common subgraphs}

Our approach to find a sufficiently good 
nice common subgraph of $F_1$ and $F_2$ 
consists in efficiently generating polynomially 
many options for the roles of the high degree root vertices 
of components of $F_1$ and $F_2$
within components of the common subgraph 
that are (potentially) of order at least $1+\frac{\Delta^2}{\epsilon}$.
Removing the corresponding parts from $F_1$ and $F_2$ yields 
forests $F_1'$ and $F_2'$, whose components all have orders 
less than $1+\frac{\Delta^2}{\epsilon}$, and Lemma \ref{lemma3} 
allows to determine common subgraphs of $F_1'$ and $F_2'$
of maximum size in polynomial time.
Returning the best overall solution encountered in this way, 
while considering the polynomially many options for the high degree roots,
yields a sufficiently good nice common subgraph of $F_1$ and $F_2$.

Let $F$ be a nice common subgraph of $F_1$ and $F_2$
that is a spanning subgraph of the forest $F_f$ 
for some bijection $f:V(F_1)\to V(F_2)$.
Let $K$ be a component of $F$ 
that contains a vertex $r$ 
with $d_{F_1}(r)\geq \frac{\Delta}{\epsilon}$
or
$d_{F_2}(f(r))\geq \frac{\Delta}{\epsilon}$.
Since $F_1$ is clean, every component of $F$ that does not contain such a vertex
necessarily has order less than $1+\frac{\Delta^2}{\epsilon}$.

It follows that
\begin{itemize}
\item $K$ is an induced subgraph of 
the component $K_1$ of $F_1$ with root $r$ and 
\item the tree $f(K)$
with vertex set $f(V(K))$ and edge set $\{ f(u)f(v):uv\in E(K)\}$
is an induced subgraph of 
the component $K_2$ of $F_2$ with root $f(r)$.
\end{itemize}
Since $F_1$ and $F_2$ are clean, 
all their components are rooted trees from ${\cal S}$.
Let $K_1$ be a copy of $S_{i_1}$ and
let $K_2$ be a copy of $S_{i_2}$.
Denote the root $r$ of $K_1$ by $r_{i_1}$ and 
the root $f(r)$ of $K_2$ by $r_{i_2}$.

For every child $u$ of $r_{i_1}$ in $K_1$, 
\begin{itemize}
\item either the edge $r_{i_1}u$ does not belong to $F$, 
which means that $u$ does not belong to $K$, 
\item or there is some child $v$ of $r_{i_2}$ in $K_2$ such that 
$r_{i_1}u$ belongs to $K\subseteq F$,
$u$ belongs to $K$,
$f(u)=v$, and
$v$ belongs to $f(K)$.
\end{itemize}
Symmetric options hold for every child $v$ of $r_{i_2}$ in $K_2$.

If some child $u$ of $r_{i_1}$ in $K_1$ belongs to $K$ and 
the child $v=f(u)$ of $r_{i_2}$ in $K_2$ belongs to $f(K)$, then 
the component of $K_1-r_{i_1}$ that contains $u$ is (the copy of) 
a tree $T_{j_1}$ from ${\cal T}$ rooted in $u$, and
the component of $K_2-r_{i_2}$ that contains $v$ is (the copy of) 
a tree $T_{j_2}$ from ${\cal T}$ rooted in $v$.
We now consider the options for the subtrees of $K$ within $K_1-r_{i_1}$
and of $f(K)$ within $K_2-r_{i_2}$.

\medskip

\medskip

\medskip

\noindent {\bf The finitely many options to overlay trees from ${\cal T}$ at their roots}

Let $X$ be the set of all $5$-tupels $x$ such that 
\begin{itemize}
\item either $x=(j_1,j_2,A_0,A_1,A_2)$, where 
\begin{itemize}
\item $j_1,j_2\in [p]$,
\item $A_0$ is one of the rooted trees from $\{ T_1,\ldots,T_p\}$, 
\item for each $\ell\in [2]$,
$A_0$ is isomorphic as a rooted tree 
to a rooted subtree $A_{0,\ell}$ of $T_{j_\ell}$
that is rooted in the root of $T_{j_\ell}$,
\item $A_1\simeq T_{j_1}-V(A_{0,1})$, and
\item $A_2\simeq T_{j_2}-V(A_{0,2})$,
\end{itemize}
see Figure \ref{fig2} for an illustration,
\item or $x=(j_1,\emptyset,\emptyset,\emptyset,\emptyset)$ for $j_1\in [p]$,
corresponding to the option that the child of $r_{i_1}$ in $K_1$
that belongs to a copy of $T_{j_1}$ in $K_1-r_{i_1}$
does not belong to $K$.
\item or $x=(\emptyset,j_2,\emptyset,\emptyset,\emptyset)$ for $j_2\in [p]$,
corresponding to the option that the child of $r_{i_2}$ in $K_2$
that belongs to a copy of $T_{j_2}$ in $K_2-r_{i_2}$
does not belong to $f(K)$.
\end{itemize}
Note that we add the $5$-tupels of the form 
$(j_1,\emptyset,\emptyset,\emptyset,\emptyset)$
and 
$(\emptyset,j_2,\emptyset,\emptyset,\emptyset)$
for notational convenience:
Together with the $5$-tupels of the form $(j_1,j_2,A_0,A_1,A_2)$,
they allow to clarify the role of all children (up to symmetry) of the root of $K_1$
as well as of all children of the root of $K_2$ within $K$.

\begin{figure}[H]
   \centering
    \begin{tikzpicture} [scale=0.4] %0.4
    \tikzstyle{point}=[draw,circle,inner sep=0.cm, minimum size=1mm, fill=black]
    \tikzstyle{point2}=[draw,circle,inner sep=0.cm, minimum size=0.5mm, fill=black]

    \node[point] (u) at (8,10)  {};
    \node at (8,10.2) [label=above right:\text{root of $T_{j_1}$ and of $A_{0,1}$}] {};
    \coordinate (r) at (8,11) {};

    \begin{scope}[shift={(8,0)}]
        \node[point] (a) at (0,7) [label=left:] {};
        \node[point] (b) at (0+2*0.5,7-2*0.866) [label=left:] {};
        \node[point] (c) at (-4*0.5,7-4*0.866) [label=left:] {};
        \coordinate (a3) at (-3*0.5,7-3*0.866) [label=left:] {};

        \node[point] (cc) at (-2*0.5,7-2*0.866) [label=left:] {};
        \node[point] (d) at (4*0.5,7-4*0.866) [label=left:] {};
        \coordinate (a4) at (3*0.5,7-3*0.866) [label=left:] {};

        \node[point] (d2) at (4*0.5,7-6*0.866) [label=left:] {};
        \node[point] (e) at (0,7-4*0.866) [label=left:] {};
        \draw (a) -- (d) -- (d2);
        \draw (a) -- (c);
        \draw (e) -- (b);
        \draw (a) -- (u);
    \end{scope}

    \begin{scope}[shift={(3,0)}]
        \node[point] (a) at (0,7) [label=left:] {};
        \coordinate (a0) at (-0.5,7.5) [label=left:] {};

        \coordinate (a2) at (0.5,7-0.866) [label=left:] {};
        \node[point] (b) at (2*0.5,7-2*0.866) [label=left:] {};
        \node[point] (cc) at (-2*0.5,7-4*0.866) [label=left:] {};
        \node[point] (c) at (-2*0.5,7-2*0.866) [label=left:] {};
        \coordinate (a1a) at (-2.3,7-3*0.866) [label=left:] {};
        \coordinate (a1b) at (-0.8,7-3*0.866) [label=left:] {};

        \node[point] (x) at (-4*0.5,7-6*0.866) [label=left:] {};
        \node[point] (y) at (0,7-6*0.866) [label=left:] {};
        \draw (x) -- (cc) -- (y);

        \node[point] (d) at (2*0.5,7-4*0.866) [label=left:] {};
        \node[point] (d2) at (2*0.5,7-6*0.866) [label=left:] {};
        \draw (a) -- (b) -- (d2);
        \draw (a) -- (c) -- (cc);
        \draw (a) -- (u);
    \end{scope}

    \begin{scope}[shift={(13,0)}]
        \node[point] (a) at (0,7) [label=left:] {};
        \coordinate (a8) at (1,7.5) [label=left:] {};

        \node[point] (b) at (2*0.5,7-2*0.866) [label=left:] {};
        \coordinate (a5) at (-1*0.5,7-1*0.866) [label=left:] {};
        \node[point] (c) at (-2*0.5,7-2*0.866) [label=left:] {};
        \node[point] (cc) at (-2*0.5,7-4*0.866) [label=left:] {};
        \node[point] (ccc) at (-2*0.5,7-6*0.866) [label=left:] {};
        \coordinate (a6) at (2*0.5,7-3*0.866) [label=left:] {};
        \coordinate (a7) at (4*0.5,7-3*0.866) [label=left:] {};

        \node[point] (d) at (4*0.5,7-4*0.866) [label=left:] {};
        \node[point] (x) at (2*0.5,7-6*0.866) [label=left:] {};
        \node[point] (y) at (6*0.5,7-6*0.866) [label=left:] {};
        \draw (x) -- (d);
        \draw (a) -- (y);
        \draw (a) -- (c) -- (cc) -- (ccc);
        \draw (a) -- (u);
    \end{scope}

    \draw [rounded corners, rounded corners=3mm] (0,1+0.1) -- ($(a1a)-(0,0.3)$) -- ($(a1b)-(0,0.3)$) -- ($(a2)-(-0.4,0)$) -- ($(a3)-(0,0.3)$) -- ($(a4)-(0,0.3)$) -- ($(a5)-(0.3,0)$) -- ($(a6)-(0.2,0.3)$) -- ($(a7)-(-0.3,0.2)$) -- (17.2,1+0.1) -- cycle;

    \node[align=left] at (18,3) {$A_1$};

    \draw [rounded corners, rounded corners=3mm] (a0) -- ($(a1a)-(0,-0.1)$) -- ($(a1b)-(0,-0.1)$) -- ($(a2)-(-0.4,-0.4)$) -- ($(a3)-(0,-0.1)$) -- ($(a4)-(0,-0.1)$) -- ($(a5)-(0.3,-0.4)$) -- ($(a6)-(0.2,-0.1)$) -- ($(a7)-(-0.3,-0.2)$) -- (a8) -- (r) -- cycle;

    \node[align=left] at (15,8) {$A_{0,1}$};

\begin{scope}[shift={(20,0)}]
    \tikzstyle{point}=[draw,circle,inner sep=0.cm, minimum size=1mm, fill=black]
    \tikzstyle{point2}=[draw,circle,inner sep=0.cm, minimum size=0.5mm, fill=black]

    \node[point] (u) at (8,10) {};
    \node at (8,10.2) [label=above right:\text{root of $T_{j_2}$ and of $A_{0,2}$}] {};
    \coordinate (r) at (8,11)  {};

    \begin{scope}[shift={(8,0)}]
        \node[point] (a) at (0,7) [label=left:] {};
        \node[point] (b) at (0+2*0.5,7-2*0.866) [label=left:] {};
        \node[point] (c) at (-4*0.5,7-4*0.866) [label=left:] {};

        \coordinate (a3) at (-2*0.5,7-3*0.866) [label=left:] {};

        \node[point] (cc) at (-2*0.5,7-2*0.866) [label=left:] {};

        \node[point] (d) at (4*0.5,7-4*0.866) [label=left:] {};
        \coordinate (a4) at (2*0.5,7-1*0.866) [label=left:] {};

        %\node[point] (d2) at (4*0.5,7-6*0.866) [label=left:] {};
        \node[point] (e) at (0,7-4*0.866) [label=left:] {};

        \node[point] (f1) at (-1,7-6*0.866) [label=left:] {};
        \node[point] (f2) at (1,7-6*0.866) [label=left:] {};
        \node[point] (f3) at (1,7-8*0.866) [label=left:] {};

        \draw (f1) -- (e) -- (f2) -- (f3);

        \draw (a) -- (d);
        \draw (a) -- (c);
        \draw (e) -- (b);
        \draw (a) -- (u);
    \end{scope}

    \begin{scope}[shift={(3,0)}]
        \node[point] (a) at (0,7) [label=left:] {};
        \coordinate (a0) at (-0.5,7.5) [label=left:] {};

        \coordinate (a2) at (0.5,7-0.866) [label=left:] {};
        \node[point] (b) at (2*0.5,7-2*0.866) [label=left:] {};
        \node[point] (cc) at (-2*0.5,7-6*0.866) [label=left:] {};
        \node[point] (c) at (-2*0.5,7-2*0.866) [label=left:] {};
        \node[point] at (-2*0.5,7-4*0.866) [label=left:] {};

        \coordinate (a1a) at (-2.3,7-3*0.866) [label=left:] {};
        \coordinate (a1b) at (-0.8,7-3*0.866) [label=left:] {};

        \node[point] (x) at (-4*0.5,7-8*0.866) [label=left:] {};
        \node[point] (y) at (0,7-8*0.866) [label=left:] {};
        \draw (x) -- (cc) -- (y);

        \node[point] (d) at (2*0.5,7-4*0.866) [label=left:] {};
        \node[point] (d2) at (2*0.5,7-6*0.866) [label=left:] {};
        \draw (a) -- (b) -- (d2);
        \draw (a) -- (c) -- (cc);
        \draw (a) -- (u);
    \end{scope}

    \begin{scope}[shift={(13,0)}]
        \node[point] (a) at (0,7) [label=left:] {};
        \coordinate (a8) at (1,7.5) [label=left:] {};

        \node[point] (b) at (2*0.5,7-2*0.866) [label=left:] {};
        \coordinate (a5) at (-2*0.5,7-3*0.866) [label=left:] {};
        \node[point] (c) at (-2*0.5,7-2*0.866) [label=left:] {};
        \node[point] (cc) at (-2*0.5,7-4*0.866) [label=left:] {};
        \node[point] (ccc) at (-2*0.5,7-6*0.866) [label=left:] {};

        \node[point] (f1) at (-4*0.5,7-8*0.866) [label=left:] {};
        \node[point] (f2) at (0*0.5,7-8*0.866) [label=left:] {};
        \draw (f1) -- (ccc) -- (f2);

        \coordinate (a6) at (2*0.5,7-3*0.866) [label=left:] {};
        \coordinate (a7) at (4*0.5,7-3*0.866) [label=left:] {};

        \node[point] (d) at (4*0.5,7-4*0.866) [label=left:] {};
        \node[point] (x) at (2*0.5,7-6*0.866) [label=left:] {};
        \node[point] (y) at (6*0.5,7-6*0.866) [label=left:] {};
        \draw (x) -- (d);
        \draw (a) -- (y);
        \draw (a) -- (c) -- (cc) -- (ccc);
        \draw (a) -- (u);
    \end{scope}

    \draw [rounded corners, rounded corners=3mm] (0,1-2*0.866+0.1) -- ($(a1a)-(0,0.3)$) -- ($(a1b)-(0,0.3)$) -- ($(a2)-(-0.4,0)$) -- ($(a3)-(0,0.3)$) -- ($(a4)-(0.2,0)$) -- ($(a5)-(0.3,0.2)$) -- ($(a7)-(-0.3,0.2)$) -- (16.9 ,7-6*0.866-0.2) -- (13,1-2*0.866+0.1) -- cycle;

    \node[align=left] at (18,3) {$A_2$};

    \draw [rounded corners, rounded corners=3mm] (a0) -- ($(a1a)-(0,-0.1)$) -- ($(a1b)-(0,-0.1)$) -- ($(a2)-(-0.4,-0.4)$) -- ($(a3)-(0,-0.1)$) -- ($(a4)-(0.2,-0.4)$) -- ($(a5)-(0.3,-0.2)$) -- ($(a7)-(-0.3,-0.2)$) -- (a8) -- (r) -- cycle;

    \node[align=left] at (15,8) {$A_{0,2}$};
\end{scope}

\end{tikzpicture}

\caption{Subgraphs 
$A_{0,1}$ and $A_1$ of $T_{j_1}\subseteq F_1$ (on the left) and
$A_{0,2}$ and $A_2$ of $T_{j_2}\subseteq F_2$ (on the right) 
for some $x=(j_1,j_2,A_0,A_1,A_2)$ in $X$.
Note that there are different isomorphic copies of $A_0$ 
within $T_{j_1}$ and $T_{j_2}$ containing their roots,
which lead to different possibilities for the subforests $A_1$ and $A_2$,
completely specified up to isomorphism by the considered $5$-tupels.}\label{fig2}
\end{figure}
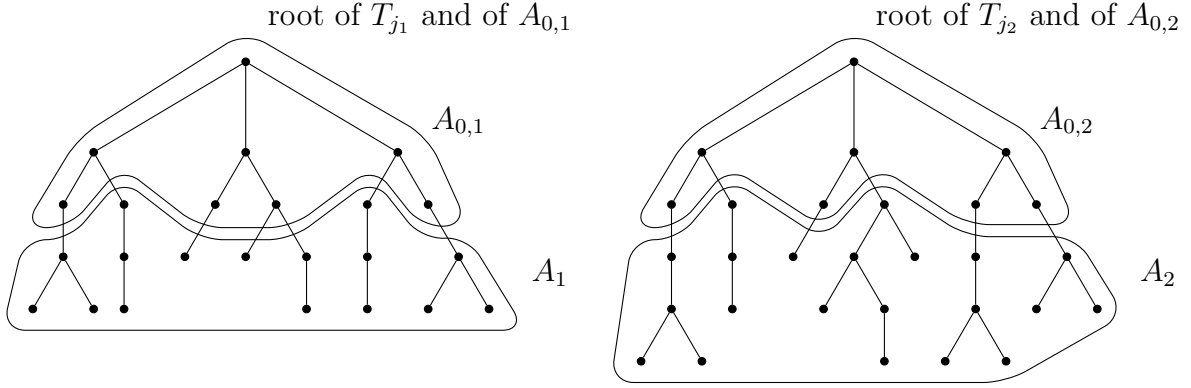
Since, for fixed $\epsilon$, 
the number $p$ of trees in ${\cal T}$ is finite and 
all these trees have finite order at most $\Delta$,
the set $X$ is finite for fixed $\epsilon$.
Let 
$$X=\{ x_1,\ldots,x_o\}.$$
For $j\in [p]$, 
let $X_{(j,*,*,*,*)}$ be the subset of all $5$-tupels in $X$ 
that have $j$ as their first entry, and
let $X_{(*,j,*,*,*)}$ be the subset of all $5$-tupels in $X$ 
that have $j$ as their second entry.

For each $\ell\in [2]$ and $j\in [p]$, 
let $S_{i_\ell}-r_{i_\ell}$ contain $t_{\ell,j}$ components of $T_j$, that is, 
the $p$-tupels $(t_{1,1},\ldots,t_{1,p})$ and $(t_{2,1},\ldots,t_{2,p})$
determine $S_{i_1}$ and $S_{i_2}$ up to isomorphism, respectively.
In particular,
$t_{\ell,1}+\cdots+t_{\ell,p}$ equals the degree $d_{i_\ell}$ 
of the root $r_{i_\ell}$ of $S_{i_\ell}$.
See Figure \ref{fig1} illustrating the structure of the rooted trees in ${\cal S}$.

\medskip

\medskip

\medskip

\noindent {\bf Quantizing the options for a single (possibly) large component}

All essentially different options for $K$ within $K_1\simeq S_{i_1}$
and $f(K)$ within $K_2\simeq S_{i_2}$ 
can be encoded in the obvious way
by an $o$-tupel $(y_1,\ldots,y_o)\in \mathbb{N}_0^o$;
in particular, 
for $x_\ell=(j_1,j_2,A_0,A_1,A_2)$ in $X$,
there are $y_\ell$ pairs $(L_1,L_2)$ such that 
$L_1$ is a component of $K_1-r_{i_1}$ isomorphic to $T_{j_1}$,
$L_2$ is a component of $K_2-r_{i_2}$ isomorphic to $T_{j_2}$,
$K$ contains the root of $L_1$,
$f(K)$ contains the root of $L_2$,
$K\cap L_1\simeq A_0$,
$f(K)\cap L_2\simeq A_0$,
$L_1-V(K)\simeq A_1$, and
$L_2-V(f(K))\simeq A_2$.
Note that
\begin{eqnarray}\label{e6}
t_{1,j} =\sum\limits_{x_\nu\in X_{(j,*,*,*,*)}}y_\nu
\,\,\,\,\,\mbox{ and }\,\,\,\,\,
t_{2,j} = \sum\limits_{x_\nu\in X_{(*,j,*,*,*)}}y_\nu
\,\,\,\,\,\mbox{ for every $j\in [p]$.}
\end{eqnarray}
Let 
$$Y_{(i_1,i_2)}$$ 
denote the set of all these $o$-tupels,
which depends only on $i_1,i_2\in [q]$. 
In principle, all $o$-tupels in $Y_{(i_1,i_2)}$ 
may be relevant for optimally solving 
the largest common subgraph problem for $F_1$ and $F_2$.
Their number though would be too large to obtain a polynomial running time.
Since we aim for an approximate solution only, 
we may restrict these options by quantization.
Therefore, let 
$\tilde{Y}_{(i_1,i_2)}$
be the set 
of all $o$-tupels $(y_1,\ldots,y_o)$ in $Y_{(i_1,i_2)}$
such that, for every $\nu\in [o]$ 
such that the first two entries of $x_\nu$ belong to $[p]$,
the value of $y_\nu$ is a multiple of 
$$\delta=\max\left\{1,
\left\lfloor\frac{\epsilon(d_{i_1}+d_{i_2})}{o\Delta}\right\rfloor
\right\}.$$
Note that, by (\ref{e6}), the $y_\nu$ with $x_\nu$ of the form 
$(j,\emptyset,\emptyset,\emptyset,\emptyset)$ 
or
$(\emptyset,j,\emptyset,\emptyset,\emptyset)$ for some $j\in [p]$,
are determined by the remaining $y_\nu$ and the $t_{\ell,j}$.
Since $y_1+\cdots +y_o\leq d_{i_1}+d_{i_2}$,
the step-size $\delta$ 
leaves at most $1+\frac{2o\Delta}{\epsilon}$ possible values for each $y_\nu$
with $\nu\in [o]$ such that the first two entries of $x_\nu$ belong to $[p]$
and, hence,
$$\left|\tilde{Y}_{(i_1,i_2)}\right|\leq 
o_{\max}:=\left(1+\frac{2o\Delta}{\epsilon}\right)^{o-2p},$$
which is finite for fixed $\epsilon$.

Enumerate the elements of this set as  
$$\tilde{Y}_{(i_1,i_2)}=
\left\{ y^1_{(i_1,i_2)},\ldots,y^{o_{(i_1,i_2)}}_{(i_1,i_2)}\right\}.$$
Restricting, for all $(i_1,i_2)\in [q]^2$,
to $\tilde{Y}_{(i_1,i_2)}$ instead of $Y_{(i_1,i_2)}$
for the approximate solution of the largest common subgraph problem 
for $F_1$ and $F_2$, deteriorates the achievable solutions by at most $2\epsilon n$.
In fact, if $\delta=1$ for $(i_1,i_2)\in [p]^2$, 
then $Y_{(i_1,i_2)}=\tilde{Y}_{(i_1,i_2)}$ 
and nothing changes.
If $\delta>1$ for $(i_1,i_2)\in [p]^2$,
then, 
for each $o$-tupel $(y_1,\ldots,y_o)$ in $Y_{(i_1,i_2)}\setminus \tilde{Y}_{(i_1,i_2)}$,
reducing each entry $y_\nu$
such that the first two entries $x_\nu$ belongs to $[p]$,
say $x_\nu=(j_1,j_2,A_0,A_1,A_2)$,
by less than $\delta$,
and increasing both entries for $(j_1,\emptyset,\emptyset,\emptyset,\emptyset)$ 
and $(\emptyset,j_2,\emptyset,\emptyset,\emptyset)$
by exactly the same amount,
yields an $o$-tupel in $\tilde{Y}_{(i_1,i_2)}$.
Since each tree in ${\cal T}$ has order at most $\Delta$,
using this new $o$-tupel instead of the old one,
reduces the number of edges from $K_1\simeq S_{i_1}$ 
in the solution by at most $\delta o \Delta\leq \epsilon(d_{i_1}+d_{i_2})$,
which is at most an $\epsilon$-fraction of 
the number of edges in $K_1\simeq S_{i_1}$ plus the  
the number of edges in $K_2\simeq S_{i_2}$;
this relative local error sums up to at most $2\epsilon n$.

Note that not all elements of $\tilde{Y}_{(i_1,i_2)}$ 
lead to a component in the solution 
that contains the roots of copies of $S_{i_1}$ (and $S_{i_2}$)
and has order at least $1+\frac{\Delta^2}{\epsilon}$,
but every such large component corresponds to an element of $\tilde{Y}_{(i_1,i_2)}$.

\medskip

\medskip

\medskip

\noindent {\bf Quantizing the options for all (possibly) large components}

Note that ${\cal T}$, ${\cal S}$, $X$, and $\tilde{Y}_{(i_1,i_2)}$ for $i_1,i_2\in [p]$
only depend on $\epsilon$ and $n$ 
but not on the specific clean forests $F_1$ and $F_2$.
Having understood and restricted the possible large components 
arising from a pair of components, one from $F_1$ and one from $F_2$,
we now consider $F_1$ and $F_2$ as a whole.

For $\ell\in [2]$ and $i\in [q]$, 
let $F_\ell$ contain $s_{\ell,i}$ components that are copies of $S_i$,
that is,
$$F_1\simeq\bigcup_{i=1}^q s_{1,i}S_i\,\,\,\,\,\mbox{ and }\,\,\,\,\,
F_2\simeq\bigcup_{i=1}^q s_{2,i}S_i.$$
Consider a nice common subgraph $F$ of $F_1$ and $F_2$
that is a spanning subgraph of the forest $F_f$ 
for some bijection $f:V(F_1)\to V(F_2)$,
such that, 
for every $(i_1,i_2)\in [q]^2$,
every component $K_1$ of $F_1$ isomorphic to $S_{i_1}$,
and 
every component $K_2$ of $F_2$ isomorphic to $S_{i_2}$
such that $f$ maps the root $r$ of $K_1$ to the root of $K_2$,
and $r$ is not isolated in $F$, 
the common subgraph $F$ is compatible on $K_1$ and $K_2$
with some element of $\tilde{Y}_{(i_1,i_2)}$.
Note that, in this case, since $r$ is not isolated,
(\ref{e7}) and the niceness of $F$ imply $|i_2-i_1|\leq c_2$.
For every $i_1\in [q]$,
every $i_2\in [q]$ with $|i_2-i_1|\leq c_2$, and
every $k\in \left[o_{(i_1,i_2)}\right]$,
let $s(i_1,i_2,k)$ be the number of components $K_1$ of $F_1$ isomorphic to $S_{i_1}$
whose root $r$ is mapped by $f$ to the root of some component $K_2$ of $F_2$ isomorphic to $S_{i_2}$ such that
$r$ is not isolated in $F$, and 
the component $K$ of $F$ that contains $r$ 
corresponds to the element $y^k_{(i_1,i_2)}$ of $\tilde{Y}_{(i_1,i_2)}$.
Note that $S_{i_1}$ has at most $d_{i_1}\Delta$ edges.
Reducing each value $s(i_1,i_2,k)$ by less than 
$$\delta'(i_1)=\max\left\{1,
\left\lfloor\frac{\epsilon s_{1,i_1}}{(2c_2+1)o_{\max}\Delta}\right\rfloor
\right\}$$
corresponds to isolating certain roots of components of $F_1$ within $F$
and deteriorates the corresponding overall solution by less than
\begin{eqnarray*}
\sum\limits_{i_1\in [q]}
\sum\limits_{i_2\in [q]:|i_2-i_1|\leq c_2}
\frac{\epsilon s_{1,i_1}}{(2c_2+1)o_{\max}\Delta}o_{(i_1,i_2)}d_{i_1}\Delta
&\leq& 
\sum\limits_{i_1\in [q]}
\sum\limits_{i_2\in [q]:|i_2-i_1|\leq c_2}
\frac{\epsilon s_{1,i_1}}{(2c_2+1)}d_{i_1}\\
&\leq &
\sum\limits_{i_1\in [q]}
\epsilon s_{1,i_1}d_{i_1}\\
&\leq &\epsilon m(F_1)\\
&\leq &\epsilon n. 
\end{eqnarray*}
Therefore, we may restrict ourselves,
for each $i_1\in [q]$, 
to $O\left(\frac{(2c_2+1)o_{\max}\Delta}{\epsilon}\right)$
different values for each $s(i_1,i_2,k)$, 
which, for fixed $\epsilon$, yields finitely many choices, 
say at most $c_3$, for 
$$M(i_1)=\Big(s(i_1,i_2,k)\Big)_{(i_2,k)\in [q]\times \left[o_{(i_1,i_2)}\right]
:|i_2-i_1|\leq c_2}.$$
Since $q\leq c_1\log(n)$, 
this results in at most 
$$c_3^{c_1\log(n)}=n^{c_1\log(c_3)}$$ 
many choices for
$$M=\Big(s(i_1,i_2,k)\Big)_{(i_1,i_2,k)\in [q]^2\times \left[o_{(i_1,i_2)}\right]
:|i_2-i_1|\leq c_2},$$
that is, polynomially many.
Note that such an $M$ is compatible with the instance $F_1$ and $F_2$ if
\begin{eqnarray*}
s_{1,i_1}-\sum\limits_{i_2\in [q]:|i_2-i_1|\leq c_2}\,\,\,\,\,\sum\limits_{k\in \left[o_{(i_1,i_2)}\right]}s(i_1,i_2,k) & \geq & 0\mbox{ for every $i_1\in [q]$, and}\\
s_{2,i_2}-\sum\limits_{i_1\in [q]:|i_2-i_1|\leq c_2}\,\,\,\,\,\sum\limits_{k\in \left[o_{(i_1,i_2)}\right]}s(i_1,i_2,k) & \geq & 0\mbox{ for every $i_2\in [q]$}.
\end{eqnarray*}
In fact, 
these two differences count the number of roots of components of 
$F_1$ and $F_2$, respectively, 
that are either of degree less than $\frac{\Delta}{\epsilon}$
or are of larger degree and correspond to an isolated vertex in the nice solution $F$.

See Figure \ref{fig4} for an illustration.

\begin{figure}[H]
    \centering
    \begin{tikzpicture} [scale=0.275]
    \tikzstyle{point}=[draw,circle,inner sep=0.cm, minimum size=1mm, fill=black]
    \tikzstyle{point2}=[draw,circle,inner sep=0.cm, minimum size=0.5mm, fill=black]

\def \n {6}
\def \radius {3cm}

\def \nn {7}
\def \radiuss {12cm}

\begin{scope}[shift={(0,-12)}]
\begin{scope}[yscale=-1,xscale=1]

\foreach \j in {0,1,2,3,4,5,6}{
\begin{scope}[shift={(8*\j,0)}]
\foreach \i in {1,...,\n}{
  \draw ({360/\n * (\i - 1)}:\radius) -- (0,0);
    \node[point] at  ({360/\n * (\i - 0.5)}:0.66*\radius) [label=below:] {};
\coordinate (x\i\j) at  ({360/\n * (\i - 0.5)}:0.66*\radius) [label=below:] {};
}
\draw (0,0) circle (3cm);
\end{scope}
}

\foreach \j in {0,1,2,3,4,5,6}{
\begin{scope}[shift={(8*\j,-12)}]       
\foreach \i in {1,...,\n}{
  \draw ({360/\n * (\i - 1)}:\radius) -- (0,0);
    \node[point] at  ({360/\n * (\i - 0.5)}:0.66*\radius) [label=below:] {};
    \coordinate (y\i\j) at  ({360/\n * (\i - 0.5)}:0.66*\radius) [label=below:] {};
}
\draw (0,0) circle (3cm);
\end{scope}
}

\end{scope}
\end{scope}

\foreach \j in {0,1,2,3,4,5,6}{

\foreach \i in {1}{
\pgfmathparse{\j+1}
\let\k\pgfmathresult
\foreach \l in {4}{
    \ifthenelse{\j<6}{\draw (y\i\j) -- (x\l\k)}{};
  }
}
\foreach \i in {2}{
\foreach \l in {5}{
    \draw (y\i\j) -- (x\l\j);
  }
}
\foreach \i in {3}{
\pgfmathparse{\j-1}
\let\k\pgfmathresult
\foreach \l in {6}{
    \ifthenelse{\j>0}{\draw (y\i\j) -- (x\l\k)}{};
  }
}

\foreach \i in {4}{
\pgfmathparse{\j-2}
\let\k\pgfmathresult
\foreach \l in {1}{
    \ifthenelse{\j>1}{\draw (y\i\j) -- (x\l\k)}{};
  }
}

\foreach \i in {6}{
\pgfmathparse{\j+2}
\let\k\pgfmathresult
\foreach \l in {3}{
    \ifthenelse{\j<5}{\draw (y\i\j) -- (x\l\k)}{};
  }
}
}

\foreach \j in {4}{

\foreach \i in {1}{
\pgfmathparse{\j+1}
\let\k\pgfmathresult
\foreach \l in {4}{
    \ifthenelse{\j<6}{\draw[line width=0.7mm] (y\i\j) -- (x\l\k)}{};
  }
}
\foreach \i in {2}{
\foreach \l in {5}{
    \draw[line width=0.7mm] (y\i\j) -- (x\l\j);
  }
}
\foreach \i in {3}{
\pgfmathparse{\j-1}
\let\k\pgfmathresult
\foreach \l in {6}{
    \ifthenelse{\j>0}{\draw[line width=0.7mm] (y\i\j) -- (x\l\k)}{};
  }
}

\foreach \i in {4}{
\pgfmathparse{\j-2}
\let\k\pgfmathresult
\foreach \l in {1}{
    \ifthenelse{\j>1}{\draw[line width=0.7mm] (y\i\j) -- (x\l\k)}{};
  }
}

\foreach \i in {6}{
\pgfmathparse{\j+2}
\let\k\pgfmathresult
\foreach \l in {3}{
    \ifthenelse{\j<5}{\draw[line width=0.7mm] (y\i\j) -- (x\l\k)}{};
  }
}
}

\draw[dashed] (12,7) -- (12,-21);

\begin{scope}[shift={(0,24)}]
\draw[->] (10,-19) -- (6,-19);
\draw[->] (14,-19) -- (18,-19);
\node at (8,-20+2) [label=above:] {$d_i< \frac{\Delta}{\epsilon}$};
\node at (16,-20+2) [label=above:] {$d_i\ge \frac{\Delta}{\epsilon}$};
\end{scope}
\draw[thick,black,decorate,decoration={brace,amplitude=10}] (11,-17) -- (-3,-17) node[midway, below,yshift=-10]{$n(S_i)< 1+\frac{\Delta^2}{\epsilon}$};

\node at (-5,0) [label=above:] {$F_1$};
\node at (-5,-12) [label=above:] {$F_2$};

\foreach \j in {4}{
\begin{scope}[shift={(8*\j,0)}]
\node[align=center] at (0,6) {$s_{1,i_1}$ many \\$S_{i_1}$};
\end{scope}
}

\foreach \j in {-1,0,+1}{
\begin{scope}[shift={(8*4+8*\j,0)}]
\ifthenelse{\j=0}{\node at (0,-17) [label=above:] {$S_{i_1}$}}{\node at (0,-17) [label=above:] {$S_{i_1 \j}$}};
\end{scope}
}

\foreach \j in {-2}{
\begin{scope}[shift={(8*4+8*\j,0)}]
\node at (0,-17) [label=above:] {$S_{i_1 -c_2}$};
\end{scope}
}

\foreach \j in {2}{
\begin{scope}[shift={(8*4+8*\j,0)}]
\node at (0,-17) [label=above:] {$S_{i_1 +c_2}$};
\end{scope}
}

\begin{scope}[shift={(-3.75,0)}]
\node[align=right] at (40,12) [label=right:\small roots of copies of $S_{i_1}$ from $F_1$]{};
\node[point, align=right] (a) at (41,10) [label=right:\small  coincide with roots of copies]{};
\node[align=right] at ($(a)+(-1+0.75,-2)$) [label=right:\small of $S_{i_1+c_2}$ from $F_2$ in $F$]{};
\node[point, align=right] (b) at (41,6) [label=right:\small are isolated vertices in $F$]{};
\end{scope}
\draw[dashed] (y54) -- (b);
\draw[dashed] ($(y64)+(1.7,-1.8)$) -- (a);

\foreach \j in {1,2,3}{
\node[point, align=right] (b) at (50+\j,-6) {};
}

\foreach \j in {1,2,3}{
\node[point, align=right] (b) at (-2-\j,-6) {};
}

\end{tikzpicture}

\caption{The figure illustrates part of the information encoded by 
$M=\Big(s(i_1,i_2,k)\Big)_{(i_1,i_2,k)}$ as a bipartite graph $G$. 
One partite set of $G$ --- shown in the upper half ---
corresponds to $F_1$ and contains a vertex for each component of $F_1$.
Similarly for the other partite set shown in the lower half corresponding to $F_2$.
In the upper partite set, 
there are $s_{1,i_1}$ vertices corresponding to copies of $S_{i_1}$.
The edges of $G$ encode 
which root vertices of components of $F_1$ can be mapped onto
which root vertices of components of $F_2$.
Since we aim for a nice common subgraph,
the edges of $G$ reflect $c_2=2$.}\label{fig4}
\end{figure}
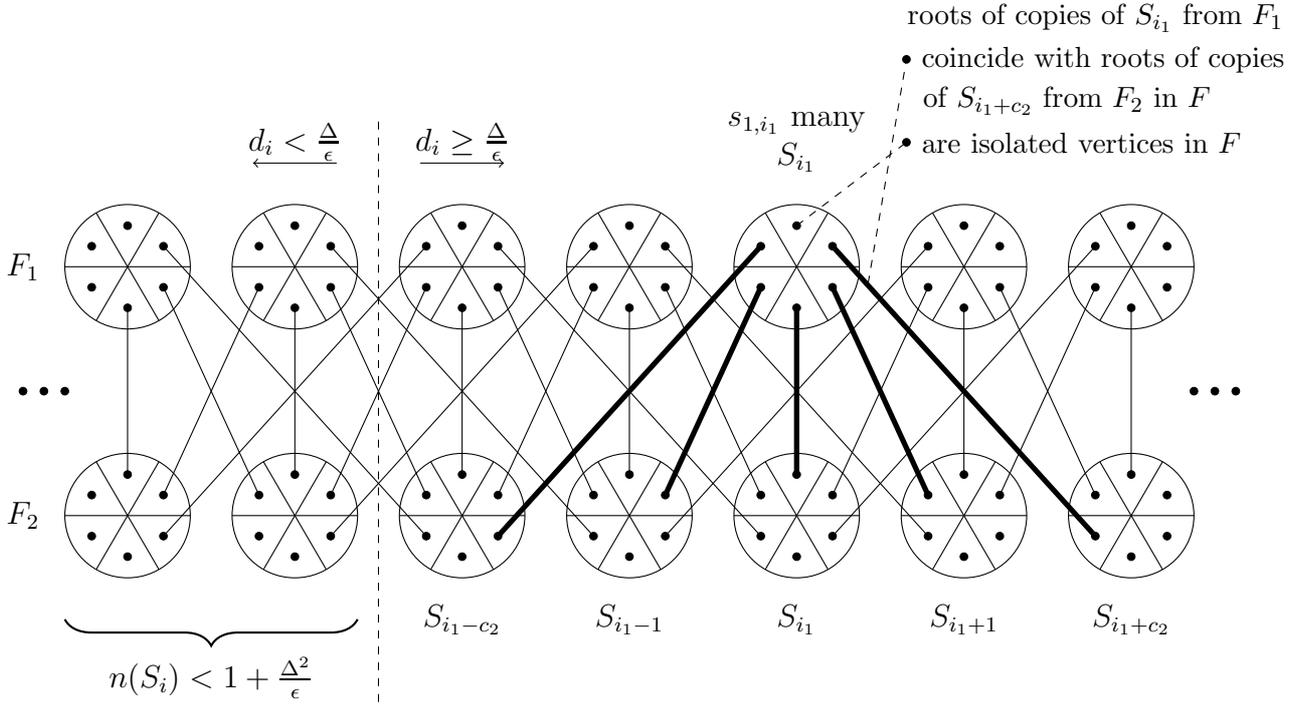

\medskip

\noindent {\bf Putting things together}

For each of the polynomially many compatible choices for $M$,  
\begin{itemize}
\item removing, 
for every $i_1,i_2\in [q]$ with $|i_2-i_1|\leq c_2$ and $k\in \left[o_{(i_1,i_2)}\right]$,
from $s(i_1,i_2,k)$ pairs $(K_1,K_2)$ of components 
$K_1\simeq S_{i_1}$ of $F_1$ 
and 
$K_2\simeq S_{i_2}$ of $F_2$
the parts $A_{0,1}$ and $A_{0,2}$ isomorphic to the corresponding subtrees $A_0$ 
as encoded by $y_{(i_1,i_2)}^k\in \tilde{Y}_{(i_1,i_2)}$, and
\item removing all edges incident with roots of degree at least $\frac{\Delta}{\epsilon}$ 
within the remaining components of $F_1$ and $F_2$,
\end{itemize}
results in subforests $F^M_1$ of $F_1$ and $F^M_2$ of $F_2$
whose components all have orders less than $1+\frac{\Delta^2}{\epsilon}$.
Let $F'$ denote the union of all parts isomorphic to the corresponding subtrees $A_0$ 
removed from $F_1$.
Using Lemma \ref{lemma3}, we can determine 
in polynomial time 
a common subgraph $F''$ of $F^M_1$ and $F^M_2$ 
with ${\rm lcs}(F^M_1,F^M_2)$ edges and $F''\subseteq F^M_1$.
Now, $F'\cup F''$ is a common subgraph of $F_1$ and $F_2$ 
that is compatible with $M$ and has the maximum possible number of edges 
subject to this condition. 
As explained along the proof, 
considering only the polynomially many choices for $M$
will produce a common subgraph $F^*$ of the form $F'\cup F''$
such that $m(F^*)\geq {\rm lcs}(F_1,F_2)-C\epsilon n$
for some fixed integer $C$ independent of $\epsilon$,
which completes the proof.
\end{proof}

\end{document}